\documentclass[journal,11pt,draftclsnofoot,onecolumn]{IEEEtran}

\usepackage{graphicx}
\usepackage[square,comma,numbers,sort&compress]{natbib}
\usepackage{amsthm,amssymb,amscd,amssymb,amsmath,mathrsfs}
\usepackage{balance}
\usepackage{xcolor}
\usepackage{overpic}
\usepackage[dvipdfmx,colorlinks,linkcolor=blue,anchorcolor=blue,citecolor=red]{hyperref}  

\makeatletter
\newcommand{\rmnum}[1]{\romannumeral #1}
\newcommand{\Rmnum}[1]{\expandafter\@slowromancap\romannumeral #1@}
\makeatother

\theoremstyle{plain}
\newtheorem{theorem}{Theorem}
\newtheorem{proposition}{Proposition}
\newtheorem{lemma}{Lemma}

\theoremstyle{definition}
\newtheorem{definition}{Definition}
\newtheorem{example}{Example}

\theoremstyle{remark}
\newtheorem{remark}{Remark}
\newtheorem{case}{Case}

\begin{document}
%
\title{{\LARGE \bf
Distributed Supervisory Control of Discrete-Event Systems with Communication Delay}
}

\author{Renyuan Zhang, Kai Cai, Yongmei Gan, W.M. Wonham
\thanks{R. Zhang is with School of Automation, Northwestern Polytechnical University, China; K. Cai is with Urban Research Plaza, Osaka City University, Japan; Y. Gan is with School of Electrical Engineering, Xi¡¯an Jiaotong University, China; and W.M. Wonham is with the Systems Control Group, Department of Electrical and Computer Engineering, University of Toronto,
Canada. (Emails: ryzhang@nwpu.edu.cn;
kai.cai@info.eng.osaka-cu.ac.jp; ymgan@mail.xjtu.edu.cn;
wonham@control.utoronto.ca). } }

\maketitle

\thispagestyle{empty} \pagestyle{plain}

\begin{abstract}
This paper identifies a property of delay-robustness in distributed
supervisory control of discrete-event systems (DES) with
communication delays. In previous work a distributed supervisory
control problem has been investigated on the assumption that
inter-agent communications take place with negligible delay. From an
applications viewpoint it is desirable to relax this constraint and
identify communicating distributed controllers which are
delay-robust, namely logically equivalent to their delay-free
counterparts. For this we introduce inter-agent channels modeled as
2-state automata, compute the overall system behavior, and present
an effective computational test for delay-robustness.  From the test
it typically results that the given delay-free distributed control
is delay-robust with respect to certain communicated events, but not
for all, thus distinguishing events which are not delay-critical
from those that are.  The approach is illustrated by a workcell
model with three communicating agents.
\end{abstract}

\section{Introduction} \label{intro:1}


Distributed control is pervasive in engineering practice, either by
geographical necessity or to circumvent the complexity of
centralized (also called `monolithic') control. Existing work on
distributed supervisory control of discrete-event systems (DES) has
focused on synthesis of local controllers for individual agents
(plant components) such that the resulting controlled behavior is
identical with that achieved by global
supervision\cite{SuThistle:2006,MannaniGohari:2008, Darondeau:2005,
SeowPham:2009, CaiWonham:2010a,CaiWonham:2010b}. In these
contributions, it is assumed that agents make independent
observations and decisions, with instantaneous inter-agent
communication. While simplifying the design of distributed control,
this assumption may be unrealistic in practice, where controllers
are linked by a physical network subject to delays. Hence, to model
and appraise these delays is essential for the correct
implementation of control strategies.

The communication problem in distributed control of multi-agent DES
has been discussed by several researchers.  Kalyon et al.
\cite{Kalyon:2011} propose a framework for the control of
distributed systems modeled as communicating finite state machines
with reliable unbounded FIFO channels. They formulate a distributed
state avoidance control problem, and show that the existence of a
solution for the problem is undecidable.  Lin\cite{Lin:2014}
investigates supervisory control of networked discrete-event systems
which features communication delays and data losses in observation
and control. He assumes that the communication between a supervisor
and the plant is via a shared network and communication delays are
bounded. Darondeau and Ricker\cite{Darondeau:2012} propose to
synthesize distributed control starting from a monolithic supervisor
(in the DES sense) which can be represented as a distributed Petri
net; local nets are linked by message passing to effect token
transfer required by transitions joining places that have been
distributed to distinct locations. PN distributability is admitted
somewhat to constrain generality; but the exact relation of this
approach to our own remains open to future research.

Research on communication problems in decentralized/modular
supervisory control has also been reported in recent years. Taking
delays into consideration, Yeddes et al. \cite{YeddesAlla:1999}
propose a 3-state data transmission model, representing delays by
timed events with lower and finite upper time bounds; these events
are incorporated into the plant and specification automata, and the
time bounds further restricted by a supervisor synthesis procedure;
maximal permissiveness and nonblocking, however, are not guaranteed.
In \cite{BarrettLafortune:2000} Barrett and Lafortune propose an
information structure model for analysis and synthesis of
decentralized supervisory control, applicable in principle to the
case of communication delays, but they assume that such delays are
absent. For a limited class of specifications,
Tripakis\cite{Tripakis:2004} formulates certain problems in
decentralized control with bounded or unbounded communication delay,
modeling the system with communication by automata with state output
map. In this model the existence of controllers in case of unbounded
delay is undecidable.  In our paper, by contrast, we address this question: does a
given controller have the property of delay-robustness (as we define
it) or not?  This question is indeed decidable, and we provide an
effective test to answer it.
Schmidt et al.\cite{Schmidt:2007} consider a heterarchical
(hierarchical/decentralized) architecture requiring communication of
shared events among modules of the hierarchy.  A communication model
is developed in which delay may affect system operation unless
suitable transmission deadlines are met.  If so, correct operation
of the distributed supervisors is achieved if the network is
sufficiently fast.  In \cite{Schmidt:2008} correct heterarchical
operation is achieved subject to a condition of ``communication
consistency'', by which the occurrence of low-level events is
restricted by the feasibility of high-level events. Xu and Kumar
\cite{Xu:2008} consider monolithic supervisory control with bounded
communication delay $d$ (measured by event count) between plant and
controller; a condition is derived for equality of controlled
behaviors under delay $d$ or with zero delay respectively;
verification is exponential in $d$. Hiraishi\cite{Hiraishi:2009}
proposes an automaton formalism for communication  with delay in
decentralized control, and concludes semi-decidability of the
controller design problem in the case of $k$-bounded delay and in
case an observability condition holds for state-transition cycles.
Ricker and Caillaud\cite{Ricker:2011} consider decentralized control
(with a priori given individual observable event subsets) in the
case where co-observability fails and therefore inter-supervisor
communication is needed for correct global supervision. The issue is
when, what, and to whom a given local supervisor should communicate;
a solution is proposed to the protocol design problem. In our paper
this question does not arise because, with supervisor localization,
we already declare who communicates what to whom, and the problem is
then to analyze our existing ideal (instantaneous) communication
scheme to see if it is still correct in the presence of delay.

Thus we consider distributed control with separately modeled
communication channels having unknown unbounded delay, imposed on an
existing distributed architecture known to be optimal and
nonblocking for zero delay. In this paper and its conference
precursor \cite{ZhaCaiWon:DR_conf12}, we start from the DES
distributed control scheme called `supervisor localization' reported
in \cite{CaiWonham:2010a,CaiWonham:2010b}, which describes a
systematic top-down approach to design distributed controllers which
collectively achieve global optimal and nonblocking supervision.
Briefly, we first synthesize a monolithic supervisor, or
alternatively a set of decentralized supervisors, assuming zero
delay; then we apply supervisor localization to decompose each
synthesized supervisor into local controllers for individual plant
components, in this process determining the set of events that need
to be communicated. Next, and central to the present paper, we
propose a channel model for event communication, and design a test
to verify for which events the system is delay-robust (as we define
it below).

The initial control problem is the standard `Ramadge-Wonham' (RW)
problem \cite{RamadgeWonham:87,WonhamRamadge:87,Wonham:2011a}. Here
the plant (DES to be controlled) is modeled as the synchronous
product of several DES agents (plant components), say ${\bf
AGENT}_1$, ${\bf AGENT}_2$, ..., that are independent, in the sense
that their alphabets $\Sigma_1$, $\Sigma_2$, ..., are pairwise
disjoint. In a logical sense these agents are linked by
specifications ${\bf SPEC}_1$, ${\bf SPEC}_2$, ..., each of which
(typically) restricts the behavior of an appropriate subset of the
${\bf AGENT}_i$ and is therefore modeled over the union of the
corresponding subfamily of the $\Sigma_i$. For each ${\bf SPEC}_j$,
a `decentralized' supervisory controller ${\bf SUP}_j$ is computed
in the same way as for a `monolithic' supervisor
\cite{RamadgeWonham:87}; it guarantees optimal (i.e. maximally
permissive) and nonblocking behavior of the relevant subfamily (the
`control scope' of ${\bf SPEC}_j$) of the ${\bf AGENT}_i$.   In
general it will turn out that the synchronous product of all the
${\bf SUP}_j$ is blocking (e.g. may cause deadlock in the overall
controlled behavior); in that case one or more additional
`coordinators' must be adjoined to suitably restrict the
decentralized controlled behavior (see \cite{CaiWonham:2010b} for an
example). Techniques for coordinator design are available in the
literature (e.g.
\cite{FengCai:2009,WongWonham:1998,HillTilbury:2006,SuSchuppen:2010})
and in this paper we take them for granted. On achieving
satisfactory decentralized control we finally `localize' each
decentralized supervisor, including the coordinator(s), if any, to
the agents that fall within its control scope; the algorithm that
achieves this is detailed in \cite{CaiWonham:2010a}, and we shall
refer to it as {\it Localize}. The result of {\it Localize} is that
each ${\bf AGENT}_i$ is equipped with local controllers, one for
each of the ${\bf SPEC}_j$ whose scope it falls within; in that
sense ${\bf AGENT}_i$ is now `intelligent' and semi-autonomous, with
controlled behavior ${\bf SUPLOC}_i$, say, while the synchronous
product behavior of all the ${\bf SUPLOC}_i$ is provably that of the
monolithic supervisor for the RW problem we began with. Autonomy of
the ${\bf SUPLOC}_i$ is qualified, in that normally the transition
structure of each ${\bf SUPLOC}_i$ will include events from various
other ${\bf AGENT}_k$ with $k \ne i$. The implementation of our
distributed control therefore requires instantaneous communication
by ${\bf AGENT}_k$ of `communication' events (when they occur, in
its private alphabet $\Sigma_k$) to ${\bf SUPLOC}_i$ so the latter
can properly update its state. Think of a group of motorists
maneuvering through a congested intersection without benefit of
external traffic control, each instead depending solely on signals
from (mostly) neighboring vehicles and on commonly accepted
protocols.  In our DES model each ${\bf SUPLOC}_i$ can disable only
its private controllable events, in $\Sigma_i$, but the logic of
disablement may well depend on observation of critical events from
certain other ${\bf AGENT}_k$ , as remarked above. It is clear that
if these communications are subject to indefinite time delay, then
control may become disrupted and the collective behavior logically
unacceptable. Our first aim is to devise a test to distinguish the
latter case from the `benign' situation where delay is tolerable, in
the sense that `logical' behavior is unaffected, even though in some
practical sense behavior might be degraded, for instance severely
slowed down\footnote{Similar issues are addressed in the literature
on `delay-insensitive' asynchronous networks; for the definition see
\cite{Udding:86} and for a useful summary \cite{Zhang:1997}.}.
This investigation would provide practitioners with
useful information to implement distributed supervisors by
communication channels: `fast' channels must be assigned for
communication of `delay-critical' events, while `slow'
channels suffice for `delay-robust' events.

In Sect.~\ref{sec:3}, we introduce the model of our
communication channel. As will be seen, there is an implicit constraint
that a channeled event (i.e. a communication event transmitted by a
channel with indefinite delay) can occur and be transmitted only
when its channel is available. This is similar to the mechanism of
``synchronous elastic circuits" or ``latency insensitive systems"
(e.g. \cite{KishCort:2008}); see Remark~\ref{rem:elastic}  below for
details. As a consequence, an uncontrollable
channeled event may or may not be blocked by its channel, the former
case being undesirable. Our second aim is to distinguish these two
cases; when an uncontrollable event is indeed blocked, we discuss
how long it can be delayed.

We proceed to a formal review of distributed control by supervisor
localization on the assumption of instantaneous inter-agent
communication. Then we introduce inter-agent communication with
delay, modeled by a separate logical channel for each delayed
communication event (i.e. channeled event). As our main result, both
a definition and a computational test are provided for
`delay-robustness' of the channeled distributed system with respect
to an arbitrary subset of communication events. In addition,  we
employ the standard algorithm for checking controllability to
identify whether or not an uncontrollable channeled event is blocked
by its channel. These issues are illustrated by a workcell model
with three communicating agents. Finally we present conclusions and
suggestions for future work.

\section{Preliminaries} \label{sec:2}

\subsection{Notation} \label{sec:2.1}


Following \cite{Wonham:2011a} we recall various standard concepts
and notation. Consider a system $\bf G$ of $n$ component DES ${\bf
G}_i = ({Q_i},{\Sigma_i},{\eta _i},{q_{i0}},{Q_{im}})$, $i \in N :=
\{1,2,...,n\}$, where ${Q_i}$ is the (finite) state set, ${\Sigma
_i}$ is the (finite) set of event labels, ${\eta _i}:{Q_i}\times
{\Sigma _i} \to {Q_i}$ is the transition (partial) function,
${q_{i0}}$ is the initial state, and $Q_{im} \subseteq Q$ is the set
of marker states. Each event set $\Sigma_i$ is partitioned as the
disjoint union ${\Sigma _i} = {\Sigma _{ic}} \cup {\Sigma _{iu}}$
where $\Sigma_{ic}$ (resp. $\Sigma_{iu}$) is the subset of
controllable (resp. uncontrollable) events for ${\bf G}_i$; the full
event set for $\bf G$ is the union $\Sigma  =  \cup \{{\Sigma _i}|i
\in N\}$.

Let $\Sigma _i^*$ denote the set of all finite strings of elements
in ${\Sigma _i}$, including the empty string $\epsilon $, and as
usual extend the transition function ${\eta _i}$  to ${Q_i} \times
\Sigma _i^*$, by defining ${\eta_i}({q_i},\epsilon ) = {q_i}$ ,
${\eta _i}({q_i},s{\sigma}) = {\eta _i}({\eta_i}({q_i},s),{\sigma})$
for all $q_i \in Q_i$, $s \in \Sigma_i^*$ and ${\sigma} \in {\Sigma
_i}$. We write $\eta_i(q_{i0},s)!$ to mean that $\eta_i(q_{i0},s)$
is defined. The {\it prefix closure} of a language $L$ over ${\Sigma
^*}$ is defined as $\overline L = \{ s \in {\Sigma ^*}|su \in
L{\rm{\ for\ some\ }}u \in {\Sigma^*}\}$. The {\it closed behavior}
and {\it marked behavior} of ${\bf G}_i$ are defined respectively by
$L({\bf G}_i) = \{ s \in \Sigma _i^*|{\eta_i}({q_{i0}},s{\rm{)! }}\}
$ and ${L_m}({\bf G}_i) = \{s \in L({\bf G}_i)|{\eta _i}({q_{i0}},s)
\in {Q_{im}}\}$.

As in \cite{CaiWonham:2010a,CaiWonham:2010b} we assume that the
${\bf G}_i$ are {\it a priori} independent, in the sense that their
alphabets $\Sigma_i$ are pairwise disjoint. The system $\bf G$
representing their combined behavior is defined to be their
synchronous product ${\bf G} = (Q,\Sigma ,\eta ,{q_0},{Q_m}) =
Sync({\bf G}_1,...,{\bf G}_n)$\footnote{We may safely assume that
the implementation $Sync$ of synchronous product is always
associative and commutative; for more on this technicality see
\cite{Wonham:2011a}, Sect. 3.3.}. The closed behavior and marked
behavior of $\bf G$ are $L({\bf G}) = ||\{ L({\bf G}_i)|i \in N\} $
and ${L_m}({\bf G}) = ||\{ {L_m}({\bf G}_i)|i \in N\}$ where $||$
denotes synchronous product of languages. Assume each ${\bf G}_i$ is
trim (i.e. reachable and coreachable); then by independence, $\bf G$
is trim, i.e., $\overline {{L_m}({\bf G})} = L({\bf G})$.

Let ${\Sigma_o} \subseteq \Sigma$ be a subset of events thought of
as `observable'. We refer the reader to \cite{Wonham:2011a} for the
formal definition of natural projection $P:{\Sigma ^*} \to
\Sigma_o^*$, DES isomorphism, ${\bf G}$-controllability, and the supremal quasi-congruence relation.
Simply stated, natural projection $P$ on a string $s\in \Sigma^*$
erases all the occurrences of $\sigma \in \Sigma$ in $s$ such that
$\sigma \notin \Sigma_o$, namely $P \sigma = \epsilon$ (the empty
string); $P$ is implemented as $Project({\bf
G},Null[\Sigma-\Sigma_o])$, which returns a (state-minimal) DES $\bf
PG$ over $\Sigma_o$ such that $L_m({\bf PG}) = PL_m({\bf G})$ and
$L({\bf PG}) = PL({\bf G})$. Two DES are isomorphic if they are
identical up to relabeling of states; ${\bf G}$-controllability is
the property required for a sublanguage of ${L_m}({\bf G})$ to be
synthesizable by a supervisory controller; while projection modulo supremal quasi-congruence produces a
(possibly nondeterministic) abstraction (reduced version) of a DES
$\bf G$, denoted $Supqc({\bf G}, Null[\Sigma-\Sigma_o])$, which
preserves observable transitions and the `observer'
property\cite{WongWonham:2004,FengWonham:2010}. As detailed in
\cite{Wonham:2011a} these operations are available in a software
implementation \cite{Wonham:2011b} and will be referred to here as
needed.

\subsection{Distributed Control without Communication Delay} \label{sec:2.2}


Next we summarize the distributed control theory (assuming zero
communication delay) reported in
\cite{CaiWonham:2010a,CaiWonham:2010b}. First suppose $\bf G$ is to
be controlled to satisfy a specification language $L_m({\bf SPEC})
\subseteq {\Sigma ^*}$ represented by a DES $\bf SPEC$. Denote by $K
\subseteq {\Sigma ^*}$ the supremal controllable sublanguage of
${L_m}({\bf G}) \cap L_m({\bf SPEC})$(for details see
\cite{Wonham:2011a}). Assume $K$ is represented by the DES $\bf
SUP$, i.e. $\bf SUP$ has closed and marked behavior
\begin{equation} \label{eq:SUP}
L({\bf SUP}) = \overline {K}, \ \ \ {L_m}({\bf SUP}) = K.
\end{equation}

Since ${\bf G} = Sync({\bf G}_1, ..., {\bf G}_n)$ is the synchronous
product of independent components we seek to implement $\bf SUP$ in
distributed fashion by `localizing' $\bf SUP$ to each ${\bf G}_i$ as
proposed in \cite{CaiWonham:2010a,CaiWonham:2010b}. For this we
bring in a family of local controllers ${\bf LOC} = \{ {\bf LOC}_i|i
\in N\}$, one for each ${\bf G}_i$, and define $L({\bf LOC}) =  \|
\{ L({\bf LOC}_i)|i \in N\} $ and ${L_m}({\bf LOC}) =  \| \{
{L_m}({\bf LOC}_i)|i \in N\}$. It is shown in
\cite{CaiWonham:2010a,CaiWonham:2010b} that
\begin{subequations} \label{e1}
    \begin{align}
        L({\bf G}) \cap L({\bf LOC}) &= L({\bf SUP})\\
        {L_m}({\bf G}) \cap {L_m}({\bf LOC}) &= {L_m}({\bf SUP})
    \end{align}
\end{subequations}
Here, the supervisory action of {\bf SUP} is fully distributed among
the set of local controllers, each acting independently and
asynchronously, except for being synchronized through
`communication' events. Generally, each local controller has a much
smaller state set than $\bf SUP$ and a smaller event subset of
$\Sigma$, containing just the events of its corresponding plant
component, together with those communication events from other
components that are essential to make correct control decisions.  We
remark that if the system and its supervisor are large scale, we
first synthesize a set of decentralized supervisors to achieve
global optimality and nonblocking, and then apply supervisor
localization to decompose each decentralized supervisor in the set
(as in \cite{CaiWonham:2010b}).

\section{Distributed Control with Communication Delay} \label{sec:3}


Cai and Wonham \cite{CaiWonham:2010a} discuss a boundary case of
optimal distributed control that is {\it fully-localizable} where
inter-agent communication is not needed, namely the alphabet of each
local controller ${\bf LOC}_i$ is simply $\Sigma_i$, so that ${\bf
LOC}_i$ observes only events in its own agent ${\bf G}_i$. In this
case, no issue of delay will arise. The more general and usual case
is that inter-agent communication is imperative.

For simplicity assume temporarily that the system $\bf G$ consists
of two components ${\bf G}_1$ and ${\bf G}_2$, and let the
monolithic supervisor $\bf SUP$ (in (\ref{eq:SUP})) be given. By
localization we compute local controllers ${\bf LOC}_1$ with event
set $\Sigma_{{\bf LOC}_1}$ and ${\bf LOC}_2$ with event set
$\Sigma_{{\bf LOC}_2}$; then the local controlled behaviors are
represented by
\begin{align}
{\bf SUP}_1 &= Sync({\bf G}_1, {\bf LOC}_1) \label{e01} \\
{\bf SUP}_2 &= Sync({\bf G}_2, {\bf LOC}_2). \label{e02}
\end{align}
Let ${\bf LOCSUP} = Sync({\bf SUP}_1, {\bf SUP}_2)$. By the localization theory of \cite{CaiWonham:2010a,CaiWonham:2010b}
we know that $L({\bf LOCSUP}) = L({\bf SUP})$ and $L_m({\bf LOCSUP}) = L_m({\bf SUP})$,
namely, the synchronized behavior of ${\bf SUP}_1$ and ${\bf SUP}_2$ agrees with that of the
monolithic control $\bf SUP$ (in (\ref{eq:SUP})).

In the general localization theory (instantaneous) inter-agent
communication is both possible and necessary, so the alphabet
$\Sigma_{{\bf LOC}_1}$ of ${\bf LOC}_1$ (resp. $\Sigma_{{\bf LOC}_2}$ of
${\bf LOC}_2$) will include elements ({\it communication events})
from $\Sigma_2$ (resp. $\Sigma_1$) as well as events from its
`private' alphabet $\Sigma_1$ (resp. $\Sigma_2$). Let $\Sigma_{com,
1}$ (resp. $\Sigma_{com, 2}$) represent the set of communication
events from $\Sigma_2$ (resp. $\Sigma_1$), i.e $\Sigma_{com, 1} =
\Sigma_{{\bf LOC}_1} - \Sigma_1$ (resp. $\Sigma_{com, 2} = \Sigma_{{\bf
LOC}_2} - \Sigma_2$); then the set of communication events in $\bf
LOCSUP$ (i.e. $\bf SUP$) is
\begin{equation} \label{e03}
\Sigma_{com} = \Sigma_{com, 1} \cup \Sigma_{com, 2}.
\end{equation}
By (\ref{e01}) and (\ref{e02}), the alphabet $\Sigma_{{\bf SUP}_1}$ of ${\bf SUP}_1$ is
\begin{align}
\Sigma_{{\bf SUP}_1} &= \Sigma_1 \cup \Sigma_{com, 1}, \label{e04}
\end{align}
and the alphabet $\Sigma_{{\bf SUP}_2}$ of ${\bf SUP}_2$ is
\begin{align}
\Sigma_{{\bf SUP}_2} &= \Sigma_2 \cup \Sigma_{com, 2}. \label{e05}
\end{align}
We say that a communication event in $\Sigma_{com, 1}$ is {\it imported}
from ${\bf G}_2$ by ${\bf LOC}_1$ (resp. $\Sigma_{com, 2}$, ${\bf G}_1$
and ${\bf LOC}_2$).

\begin{remark} \label{rem0}
For every state $x$ of each controller ${\bf LOC}_i$ ($i \in N$), and
each communication event $\sigma$ in ${\bf LOC}_i$ but imported from
some other component ${\bf G}_j$ ($j \neq i$), if $\sigma$ is not
defined at $x$, we add a $\sigma$-selfloop, i.e. transition
$(x,\sigma,x)$ to ${\bf LOC}_i$. Now, $\sigma$ is defined at every
state of ${\bf LOC}_i$. With this modification, the new local
controllers ${\bf LOC}_i$ are also control equivalent to $\bf SUP$
(because ${\bf LOC}_i$ does not disable events $\sigma$ from other
components ${\bf G}_j$ and $\sigma$ will be disabled by ${\bf LOC}_j$ if
and only if it is disabled by $\bf SUP$) and the definition of
$\sigma$ at every state of ${\bf LOC}_i$ is consistent with the
assumption that ${\bf LOC}_i$ may receive $\sigma$ after indefinite
communication delay.
\end{remark}

Next we model the way selected communication events are imported
with indefinite time delay and call such events {\it channeled
events}. Let $\Sigma_{ch}$ represent the set of channeled events;
then $\Sigma_{ch} \subseteq \Sigma_{com}$ ($\Sigma_{com}$ is defined
in (\ref{e03})). For example assume that communication event $r$ in
$\Sigma_2$ is transmitted to ${\bf LOC}_1$ from ${\bf G}_2$ via a
channel modeled as the (2-state) DES ${\bf CH}(2,r,1)$ in
Fig.~\ref{fig1}\footnote{Communications among local supervisors can
be modeled in different ways, e.g.
\cite{BarrettLafortune:2000,Tripakis:2004,ParkCho:2007}. In our
model channel capacity (for each separate channeled event) is
exactly 1 (event), imposing the constraint that a given labeled
event cannot be retransmitted unless its previous instance has been
received and acknowledged by the intended recipient (see footnote
\ref{fnote4}); this constraint may not be appropriate in all
applications. We adopt this model because its structure is
reasonable, simple, and renders the distributed control problem
(with unbounded communication delay) tractable.};
then $r$ is a channeled event. In the transition structure of ${\bf
LOC}_1$, hence also of ${\bf SUP}_1$, we replace every instance of
event $r$ with a new event $r'$, the `output' of ${\bf CH}(2,r,1)$
corresponding to input $r$ (we call $r'$ the {\it signal event} of $r$);
call these modified models ${\bf LOC}_1'$, ${\bf SUP}_1'$.
Thus if and when $r$ happens to occur (in ${\bf G}_2$) ${\bf
CH}(2,r,1)$ is driven by synchronization from its initial state 0
into state 1; on the eventual (and spontaneous) execution of event
$r'$ in ${\bf SUP}_1'$, which resets ${\bf CH}(2,r,1)$ to state 0,
the execution of $r'$ will be forced by synchronization in ${\bf
LOC}_1'$. In the standard untimed model of DES
employed here, the `time delay' between an occurrence of $r$ and
$r'$ is unspecified and can be considered unbounded; indeed,
nothing in our model so far implies that $r'$ will cause an actual state
change (as opposed to selfloop) because, subsequent to the occurrence of $r$ in ${\bf G}_2$, ${\bf
SUP}_1'$ might conceivably move to states (by events other than
$r'$) where $r'$ is a selfloop and its occurrence
will not cause a state change in ${\bf SUP}_1'$. As a convention, the
control status of $r'$ (controllable or uncontrollable) is taken to be that
of $r$. Suppose in particular that $r$ in $\Sigma_2$ is
controllable. Since ${\bf LOC}_1$ has `control authority' only over
controllable events in its private alphabet $\Sigma_1$, ${\bf
LOC}_1'$ never attempts to disable $r'$ directly; $r'$ can only be
disabled implicitly by the `upstream' disablement by ${\bf LOC}_2$
of $r$.
\begin{figure}[!t]
\centering
  \includegraphics[scale=0.6]{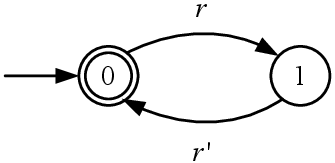}\\
  \caption{Communication channel {${\bf CH}(2,r,1)$}, from agent ${\bf G}_2$ to local controller ${\bf
  LOC}_1$ with channeled event $r$
(in the transition diagram of a DES, the circle with $\rightarrow$
represents the initial state and a double circle represents a marker
state).  One may think of the delay of $r'$ as being the \emph{sum}
of the delay of (forward) event transmission plus the delay of
(backward) acknowledgement, i.e. two delays lumped into one.
Note that when event $r$ is communicated to multiple
local controllers, we employ separate channels with distinct signal
events, as illustrated in Fig.~\ref{fig:multichn} below.}
  \label{fig1}
\end{figure}

In general ${\bf LOC}_1'$ `knows' that $r$ has occurred in ${\bf
G}_2$ only when it executes $r'$; meanwhile, other events may have
occurred in ${\bf G}_2$. The only constraint placed on events in
${\bf G}_2$ is that $r$ cannot occur again until $r'$ has finally
reset ${\bf CH}(2,r,1)$ and the communication cycle is ready to
repeat. In other words, event $r$ will be delayed in re-occurring
until the channel used to transmit event $r$ again becomes available.
If event $r$ is controllable, it can be disabled or delayed by the
local controller ${\bf LOC}_2$;\footnote{\label{fnote4}Our
model implicitly assumes that the sender (i.e. ${\bf LOC}_2$) may observe which of the
two states ${\bf CH}(2,r,1)$ is at. If ${\bf CH}(2,r,1)$ is at state 1
(the channel is not available), ${\bf LOC}_2$ disables $r$; otherwise
$r$ is enabled. In a more fine-grained model we may set $r' = r_{21}'r_{12}'$
where $r_{21}'$ signals to ${\bf LOC_1'}$ the occurrence of $r$ in ${\bf G_2}$,
while $r_{12}'$ represents an acknowledgement to ${\bf LOC_2}$ that $r_{21}'$
has occurred in ${\bf SUP_1'}$. We prove in Appendix~\ref{app0} that these two channel models
are equivalent as far as the unbounded delay-robust property is concerned.}
 but if event $r$ is uncontrollable, the constraint
placed on ${\bf G}_2$ will require that $r'$ should reset ${\bf CH}(2,r,1)$
before $r$ is enabled to occur again, possibly in violation of the
intended meaning of `uncontrollable'. This issue will be discussed in
Sect.~\ref{sec:3.3}. The channel ${\bf CH}(2,r,1)$ is not considered a
control device, but rather an intrinsic component of the physical
system being modeled; it will be `hard-wired' into the model by
synchronous product with ${\bf G}_1$ and ${\bf G}_2$.

\begin{remark} \label{rem:elastic}
We note that our model of communication channel (Fig.~\ref{fig1})
is similar to the mechanism of ``synchronous elastic circuits" or
``latency insensitive systems" (e.g. \cite{KishCort:2008}). A synchronous
elastic circuit is one whose behavior does not change despite
latencies (i.e. delays) of communication channels. One method to
build synchronous elastic circuits is ``synchronous elastic flow"
\cite{KishCort:2008}, where the idea of ``back pressure" is used in
a similar way to the ``signal events" we use in our model of communication
delay.
\end{remark}

Continuing with this special case we consider the joint behavior of ${\bf G}_1$, ${\bf G}_2$ and ${\bf CH}(2,r,1)$ under control of ${\bf LOC}_1'$ and ${\bf LOC}_2$, namely
\begin{align}
{\bf SUP}':&=Sync({\bf G}_1,{\bf LOC}_1', {\bf CH}(2,r,1), {\bf G}_2, {\bf LOC}_2) \notag\\
&=Sync({\bf SUP}_1', {\bf CH}(2,r,1), {\bf SUP}_2)\label{e1a}
\end{align}
defined over the alphabet ${\Sigma_1} \cup \{ r'\}  \cup {\Sigma_2}$. We refer to ${\bf SUP}'$ as the {\it channeled behavior} of $\bf SUP$ (in (\ref{eq:SUP})) with $r$ being the channeled event (i.e. $\Sigma_{ch} = \{r\}$).

\subsection{Delay-robustness and Delay-criticality} \label{sec:3.1}


In this subsection we formalize the definition and present an
effective computational test for delay-robustness.

Of principal interest is whether or not the communication delay
between successive occurrences of $r$ and $r'$ is tolerable in the
intuitive sense indicated above.

Let $\Sigma_{sig}$ be the set of new events introduced by the
communication channels, in which each element is the signal event of
an event in $\Sigma_{ch}$, i.e.
\begin{align} \label{e1b}
\Sigma_{sig} = \{\sigma'| \sigma \in \Sigma_{ch}, \sigma' ~\text{is the signal event of }~\sigma\}.
\end{align}
In ${\bf SUP}'$ (in (\ref{e1a})), $\Sigma_{ch} = \{r\}$ and $\Sigma_{sig} = \{r'\}$. Then the event set of ${\bf SUP}'$ will be $\Sigma' = \Sigma \cup \Sigma_{sig} = \Sigma \cup \{r'\}$. Let $P:\Sigma'^* \rightarrow \Sigma^*$ be the natural projection of $\Sigma'^*$ onto $\Sigma^*$\cite{Wonham:2011a}, i.e. $P$ maps $r'$ to $\epsilon$ (empty string).

To define whether or not ${\bf SUP}'$ with alphabet $\Sigma'$ has the same behavior as $\bf SUP$, when viewed through $P$, we require that

1. anything $\bf SUP$ can do is the $P$-projection of something ${\bf SUP}'$ can do (${\bf SUP}'$ is `complete'); and

2. no $P$-projection of anything ${\bf SUP}'$ can do is disallowed by $\bf SUP$ (${\bf SUP}'$ is `correct').

For completeness we need at least the inclusions
\begin{align}
PL({\bf SUP}') &\supseteq L({\bf SUP}) \label{e2a}\\
PL_m({\bf SUP}') &\supseteq L_m({\bf SUP}) \label{e2b}
\end{align}

In addition, however, we need the following {\it observer property} of $P$ with respect to ${\bf SUP}'$ and $\bf SUP$. Suppose ${\bf SUP}'$ executes string $s\in L({\bf SUP}')$, which will be viewed as $Ps \in L({\bf SUP})$. As $\bf SUP$ is nonblocking, there exists $w \in \Sigma^*$ such that $(Ps)w \in L_m({\bf SUP})$. For any such $w$ `chosen' by $\bf SUP$, completeness should require the ability of ${\bf SUP}'$ to provide a string $v\in \Sigma'^*$ with the property $Pv = w$ and $sv \in L_m({\bf SUP}')$. Succinctly (cf. \cite{Wonham:2011a,FengWonham:2010})
\begin{align}\label{e2c}
(\forall s\in \Sigma'^*)(\forall w \in \Sigma^*)~&s \in L({\bf SUP}')~ \& ~(Ps)w \in L_m({\bf SUP})\notag\\
\Rightarrow &(\exists v \in \Sigma'^*)~Pv = w ~\&~ sv \in L_m({\bf SUP}').
\end{align}

\begin{remark}\label{rem1}
In (\cite{Wonham:2011a}, Chapt. 6), $P$ is defined to be an $L_m({\bf SUP}')${\it -observer} if
\begin{align*}
(\forall s\in \Sigma'^*)(\forall w \in \Sigma^*)~&s \in L({\bf SUP}') ~\&~ (Ps)w \in PL_m({\bf SUP}')\notag\\
\Rightarrow &(\exists v \in \Sigma'^*)~Pv = w ~\&~ sv \in L_m({\bf SUP}').
\end{align*}
It is clear that when $PL_m({\bf SUP}') = L_m({\bf SUP})$, the observer property of $P$ with respect to ${\bf SUP}'$ and $\bf SUP$ is identical with the $L_m({\bf SUP}')$-observer property of $P$.
\end{remark}

Briefly, we define ${\bf SUP'}$ to be {\it complete} relative to $\bf SUP$ if (\ref{e2a}), (\ref{e2b}) and (\ref{e2c}) hold.

Dually, but more simply, we say that ${\bf SUP}'$ is  {\it correct} relative to $\bf SUP$ if
\begin{align}
PL({\bf SUP}') &\subseteq L({\bf SUP}) \label{e3a}\\
PL_m({\bf SUP}') &\subseteq L_m({\bf SUP}) \label{e3b}
\end{align}

To summarize, we make the following definition.

\begin{definition} \label{def1}
For given ${\bf SUP}'$ in (\ref{e1a}) and $\Sigma_{ch} = \{r\}$, $\bf
SUP$ (in (\ref{eq:SUP})) is {\it delay-robust} relative to
$\Sigma_{ch}$ provided ${\bf SUP}'$ is complete and correct relative
to $\bf SUP$, namely, conditions (\ref{e2a})-(\ref{e3b}) hold, or
explicitly
\begin{align}
&PL({\bf SUP}') = L({\bf SUP}) \label{e4a}\\
&PL_m({\bf SUP}') = L_m({\bf SUP}) \label{e4b}\\
&P~\text{has the observer property (\ref{e2c}) with respect to}~ {\bf SUP}'~\text{and}~{\bf SUP}. \tag{\ref{e2c}bis}
\end{align}
\end{definition}

We stress that in Definition~\ref{def1} (and its generalizations
later) the natural projection $P$ is fixed by the choice of
channeled events and structure of the communication model.  If the
definition happens to fail (for instance if the observer property
fails), the only cure in the present framework is to alter the set
of channeled events, in the worst case reducing it to the empty set,
that is, declaring that all communication events must be transmitted
without delay.

The following example shows why the observer property is really
needed; for if (\ref{e4a}) and (\ref{e4b}) hold, but (\ref{e2c})
fails, ${\bf SUP}'$ may have behavior which is distinguishable from
that of $\bf SUP$.

\begin{example} \label{exmp1}
Let ${\bf SUP}_1$ and ${\bf SUP}_2$ be the generators shown in
Fig.~\ref{fig10}; assume event 20 in ${\bf SUP}_2$ is exported to
${\bf SUP}_1$, i.e., $r = 20$ and $r' = 120$; ${\bf SUP}_1'$ is
obtained by replacing $20$ in ${\bf SUP}_1$ by $120$, and ${\bf SUP}'$
is obtained by (\ref{e1a}). By inspection of Fig.~\ref{fig11},
(\ref{e4a}) and (\ref{e4b}) are verified to hold. However, we can
see that (\ref{e2c}) fails. Let $s = 20.10.120.12 \in L({\bf SUP}')$;
then $Ps = 20.10.12$. Now $(Ps).11 = 20.10.12.11 \in L_m({\bf
SUP})$; but there does not exist a string $v$ such that $Pv = 11$
and $sv \in L_m({\bf SUP}')$. Thus, $\bf SUP$ can execute 11 after
$Ps$, but ${\bf SUP}'$ can only execute $\epsilon$ after $s$. This
means that ${\bf SUP}'$ has behavior distinguishable from that of $\bf
SUP$.
\end{example}
\begin{figure}[!t]
\centering
    \includegraphics[scale=0.5]{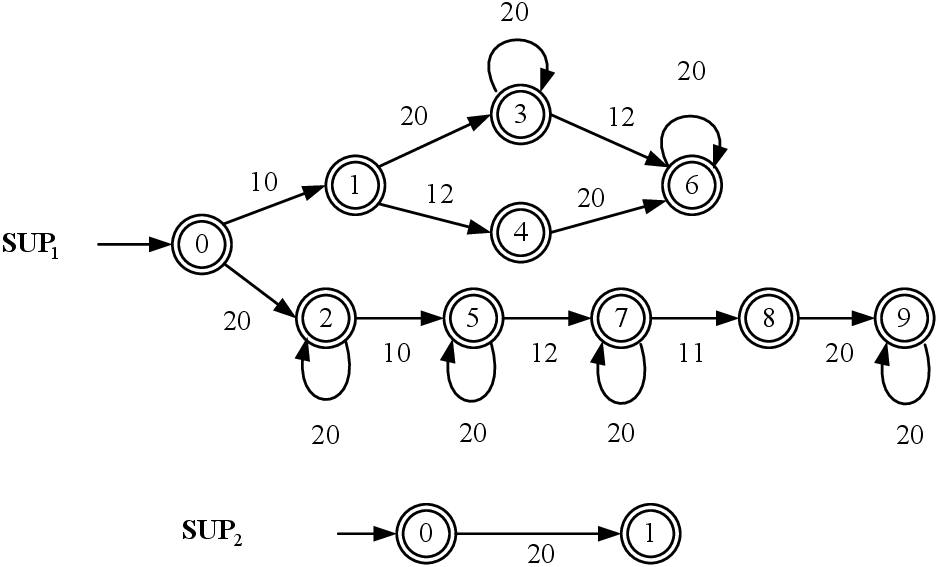}\\
  \caption{Example~\ref{exmp1}: ${{\bf SUP}_1}$ and ${{\bf SUP}_2}$}
  \label{fig10}
\end{figure}

\begin{figure}[!t]
\centering
    \includegraphics[scale=0.5]{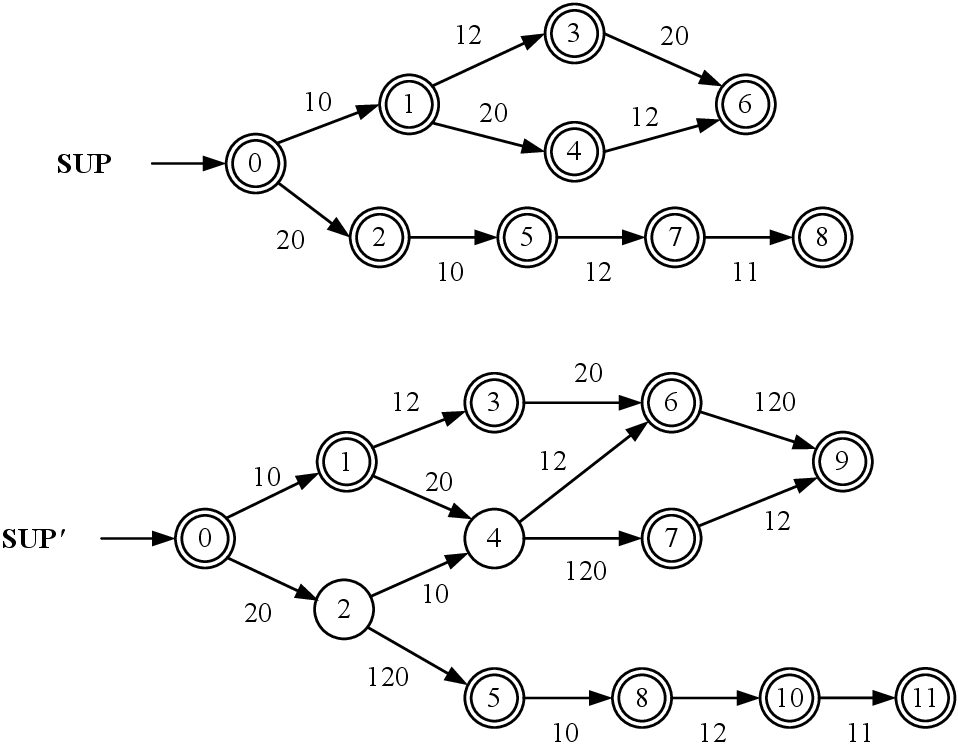}\\
  \caption{Example~\ref{exmp1}: ${\bf SUP}$ and ${\bf SUP}'$}
  \label{fig11}
\end{figure}

Since $\bf SUP$ is a nonblocking supervisor, delay-robustness of
$\bf SUP$ also requires that ${\bf SUP}'$ be nonblocking, i.e.
\begin{equation} \label{e5}
\overline{L_m({\bf SUP'})} = L({\bf SUP}'),
\end{equation}
as can easily be derived from (\ref{e2c}),(\ref{e4a}) and
(\ref{e4b}). The following example shows that when delay-robustness
fails, transmission delay of $r$ can lead to blocking in ${\bf SUP}'$.

\begin{figure}[!t]
\centering
    \includegraphics[scale=0.5]{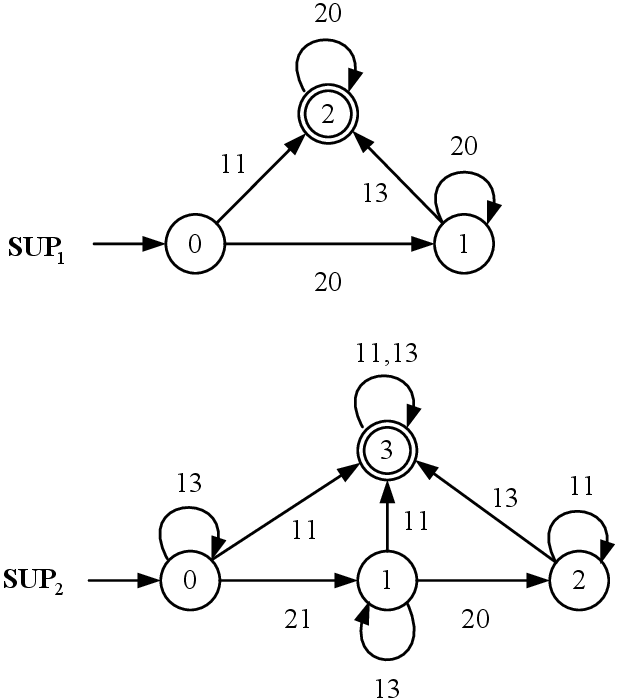}\\
  \caption{Example~\ref{exmp2}: ${{\bf SUP}_1}$ and ${{\bf SUP}_2}$}
  \label{fig2}
\end{figure}


\begin{example} \label{exmp2}
Let ${\bf SUP}_1$ and ${\bf SUP}_2$ be the generators shown in
Fig.~\ref{fig2}, and assume event 20 in ${\bf SUP}_2$ is exported to
${\bf SUP}_1$, i.e., $r = 20$ and $r' = 120$; ${\bf SUP}_1'$ is
obtained by replacing $20$ in ${\bf SUP}_1$ by $120$.  Then $\bf
SUP$ is nonblocking, but ${\bf SUP}'$ obtained by (\ref{e1a}) is
blocking, as shown in Fig. ~\ref{fig3}.  Note that delay-robustness
fails because (\ref{e4a}) fails. Indeed, string $21.20.11 \in L({\bf
SUP}')$ but $P(21.20.11)=21.20.11 \notin L({\bf SUP})$.
\begin{figure}[!t]
\centering
    \includegraphics[scale=0.5]{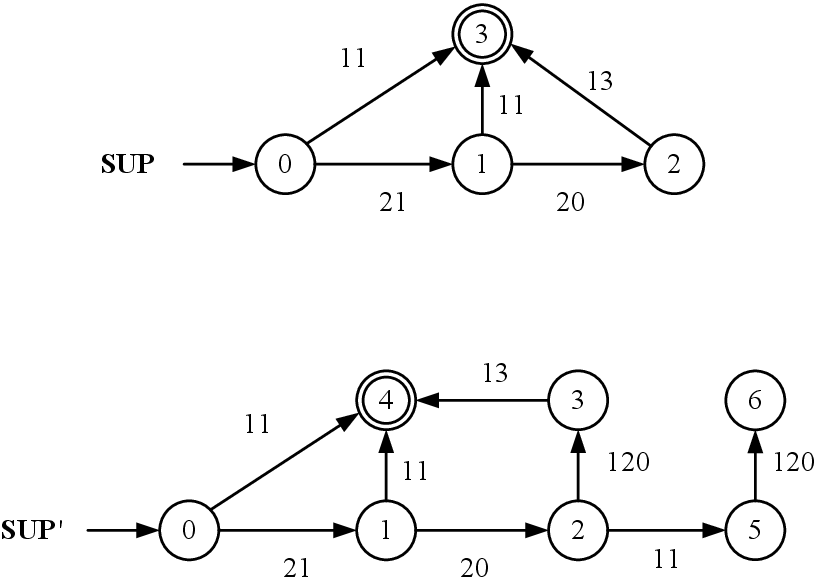}\\
  \caption{Example~\ref{exmp2}: ${\bf SUP}$ and ${\bf SUP}'$}
  \label{fig3}
\end{figure}
To see why ${\bf SUP}'$ is blocking, start from the initial state,
and suppose events 21 and 20 have occurred in ${\bf SUP}_2$ but that
${\bf SUP}_1'$ has not executed the corresponding event 120. Then
${\bf SUP}_1'$ may execute event 11, which is immediately observed
by ${\bf SUP}_2$; however, if $11$ occurs, ${\bf SUP}_1'$ and ${\bf
SUP}_2$ cannot accomplish their task synchronously; hence the system
blocks.
\end{example}

Given $\bf SUP$, $\Sigma_{ch}$, $\Sigma_{sig}$ and ${\bf SUP}'$, we
wish to verify whether or not $\bf SUP$ is delay-robust relative to
$\Sigma_{ch}$. For this we need the concept of ``supremal
quasi-congruence" \cite{Wonham:2011a,WongWonham:2004} and the
operator $Supqc$ \cite[Sect. 6.7]{Wonham:2011a} which projects a
given $\bf DES$ over the alphabet $\Sigma'$ to $\bf QCDES$, the
corresponding quotient $\bf DES$ over $\Sigma^* = P(\Sigma'^*)$. We
denote the counterpart computing procedure by \[{\bf QCDES} =
Supqc({\bf DES}, Null[])\] where $Null[]$ is the event subset
$\Sigma' - \Sigma$ that $P$ maps to the empty string $\epsilon$; for
details see \cite{Wonham:2011a}\footnote{This procedure can also be
phrased in terms of `bisimulation equivalence'\cite{Milner:89}, as
explained in \cite{WongWonham:2004}. We remark that the algorithm
for $Supqc({\bf DES}, \cdot)$ in \cite{Wonham:2011a}, Sect. 6.7, can
be estimated to have time complexity $O(kn^4)$ where $(k,n)$ is the
(alphabet, state) size of \textbf{DES}. We note that \cite{Bravo:2012}
reports an algorithm with quadratic time complexity for verifying
the observer property alone.}.
Let ${\bf QCDES} = (Z,\Sigma,\zeta,z_0, Z_m)$. In general $\bf QCDES$ will be
nondeterministic with transition function $\zeta: Z \times \Sigma^*
\rightarrow Pwr(Z)$ and include silent ($\epsilon-$) transitions. If
no silent or nondeterministic transitions happen to appear in $\bf
QCDES$, the latter is said to be `structurally deterministic'.
Formally, $\bf QCDES$ is {\it structurally deterministic} if, for
all $z\in Z$ and $s \in \Sigma^*$, we have
\[\zeta(z, s) \neq \emptyset \Rightarrow |\zeta(z, s)| = 1.\]

It is known that structural determinism of  $\bf QCDES$ is
equivalent to the condition that $P$ is an $L_m({\bf DES})$-observer
(cf. \cite{WongWonham:2004}, and \cite{Wonham:2011a}, Theorem
6.7.1).

Given minimal-state deterministic generators $\bf A$ and $\bf B$
over the same alphabet, we write ${\bf A} \subseteq {\bf B}$ iff
$L_m({\bf A}) \subseteq L_m({\bf B})$  and $L({\bf A}) \subseteq
L({\bf B})$; and ${\bf A} \approx {\bf B}$ to mean both $({\bf A}
\subseteq {\bf B})$ and $({\bf B} \subseteq {\bf A})$, i.e. ${\bf
A}$ and ${\bf B}$ are isomorphic. Clearly, ``$\approx$'' is
transitive.

Now let ${\bf SUP} = (X, \Sigma, \xi, x_0, X_m)$ (in (\ref{eq:SUP})),
${\bf SUP'} = (Y, \Sigma', \eta, y_0, Y_m)$ (in (\ref{e1a})),
\begin{align}
{\bf PSUP}' &= Project({\bf SUP}', Null[r']) \label{PSUP'}\\
{\bf QCSUP}' &= Supqc({\bf SUP}', Null[r']) \label{QCSUP'}.
\end{align}
Write ${\bf QCSUP}' = (\overline{Y}, \Sigma, \overline{\eta},
\overline{y}_0, \overline{Y}_m)$.

The following theorem provides an effective test for whether or not the communication delay is tolerable, i.e., $\bf SUP$ is delay-robust.

\begin{theorem}\label{thm1}
$\bf SUP$ is delay-robust relative to $\Sigma_{ch}$ ($ = \{r\}$) if and
only if ${\bf QCSUP}'$ is structurally deterministic, and isomorphic to $\bf SUP$.
\end{theorem}

As indicated above, ${\bf QCSUP}'$ can be computed by $Supqc$ and
isomorphism of DES can be verified by $Isomorph$.\footnote{ For
language equality {\it Isomorph} should be applied to minimal
(Nerode) state DES; see e.g. \cite{Wonham:2011a} Sect. 3.7.} Hence,
Theorem~\ref{thm1} provides an effective computational criterion for
delay-robustness. Before Theorem~\ref{thm1} is proved, a special
relation between ${\bf QCSUP}'$ and ${\bf PSUP}'$ must be
established; a proof is in Appendix~\ref{appA}.

\begin{proposition} \label{pro1}
If ${\bf QCSUP}'$ is structurally deterministic, then it is a canonical (minimal-state) generator for $PL_m({\bf SUP}')$.
\end{proposition}

\begin{proof}[Proof of Theorem~\ref{thm1}]
(If) From Proposition~\ref{pro1}, ${\bf QCSUP}'$ is a minimal state
generator of $PL_m({\bf SUP}')$. So, ${\bf QCSUP}' \approx {\bf
PSUP}'$. As ${\bf QCSUP}'$ is isomorphic to $\bf SUP$, ${\bf QCSUP}'
\approx {\bf SUP}$. Hence, ${\bf SUP} \approx {\bf PSUP}'$, i.e.
(\ref{e4a}) and (\ref{e4b}) both hold. For (\ref{e2c}), since ${\bf
QCSUP}'$ is structurally deterministic\cite[Theorem 6.7.1]{Wonham:2011a},
$P$ is an $L_m({\bf SUP}')$-observer; by Remark~\ref{rem1}
and (\ref{e4b}), $P$ has the observer property with respect to ${\bf
SUP}'$ and $\bf SUP$.  Thus by Definition~\ref{def1}, $\bf SUP$ is
delay-robust relative to $\Sigma_{ch}$.

(Only if)  By Remark~\ref{rem1}, conditions (\ref{e2c}) and
(\ref{e4b}) imply that $P$ is an $L_m({\bf SUP}')$-observer; thus
${\bf QCSUP}'$ is deterministic\cite{Wonham:2011a}. By
Proposition~\ref{pro1}, ${\bf QCSUP}' \approx {\bf PSUP}'$.
Equations (\ref{e4a}) and (\ref{e4b}) say that ${\bf PSUP}' \approx
{\bf SUP}$. Hence ${\bf QCSUP}' \approx {\bf SUP}$. Finally, we
conclude that ${\bf QCSUP}'$ is isomorphic to $\bf SUP$.
\end{proof}

\begin{remark} \label{rem:3 state channel}
In our 2-state channel model ${\bf CH}(2,r,1)$, the delay of (forward)
event transmission and the delay of (backward) acknowledgement are
lumped into one, as represented by $r'$. Here we consider a 3-state channel
model ${\bf TCH}(2,r,1)$, as shown in Fig.~\ref{fig:3stateChn}, where $r_{21}'$
signals to ${\bf LOC_1}$ the occurrence of $r$ in ${\bf G_2}$, while $r_{12}'$
represents an acknowledgement to ${\bf LOC_2}$ that ${\bf LOC_1}$ has received
the occurrence of $r$. We show in the following that: if $\bf SUP$ is delay-robust
relative to $r$ with respect to ${\bf CH}(2,r,1)$, then $\bf SUP$ is delay-robust
relative to $r$ with respect to ${\bf TCH}(2,r,1)$.

\begin{figure}[!t]
\centering
  \includegraphics[scale=0.6]{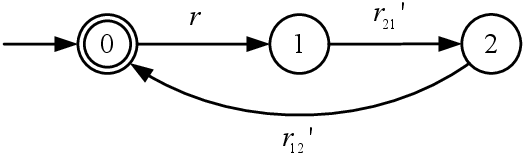}\\
  \caption{3-state Communication Channel Model ${\bf TCH}(2,r,1)$}
  \label{fig:3stateChn}
\end{figure}

Here in the transition structure of ${\bf LOC}_1$, hence also of ${\bf SUP}_1$,
we replace every instance of event $r$ with $r_{21}'$; call these modified
models ${\bf TLOC}_1'$ and ${\bf TSUP}_1'$. If and when $r$ happens to occur,
${\bf TCH}(2,r,1)$ is driven by synchronization from its initial state 0 into
state 1; the execution of event $r_{21}'$ represents that ${\bf TSUP}_1'$ has
`known' the occurrence of $r$, and the channel is brought into state 2 by synchronization;
the execution of $r_{12}'$ acknowledges that ${\bf TSUP}_1'$ has received the
occurrence of $r$ and resets the channel.

Now, the channeled behavior of the system with respect to the channel ${\bf TCH}(2,r,1)$
is
\begin{align}
{\bf TSUP}' &= Sync({\bf TSUP}_1', {\bf TCH}(2,r,1), {\bf SUP}_2)
\end{align}
and its alphabet is $\Sigma_T'= \Sigma \cup \{r_{21}',r_{12}'\}$. We prove in Appendix
\ref{app0} that:
\begin{proposition} \label{pro:relchn}
$\bf SUP$ is delay-robust
relative to $r$ with respect to ${\bf CH}(2,r,1)$, iff $\bf SUP$ is delay-robust
relative to $r$ with respect to ${\bf TCH}(2,r,1)$.
\end{proposition}
\end{remark}

We have now obtained an effective tool to determine whether or not $\bf SUP$ is delay-robust relative to $\Sigma_{ch} = \{r\}$. If $\bf SUP$ is not delay-robust relative to $r$, we say that $r$ is {\it delay-critical} for $\bf SUP$. In that case, communication of $r$ (with delay, as $r'$) could result in violation of a specification. If $r$ is delay-critical, and if such violation is inadmissible, then $r$ must be transmitted instantaneously to the agent (in this case, ${\bf LOC}_1$) that imports it -- where ``instantaneous" must be quantified on the application-determined time scale.

\subsection{Delay-robustness for Multiple Events} \label{sec:3.2}

In this subsection, we consider delay-robustness for multiple
events. First, we adopt the result of Theorem~\ref{thm1} as the
basis of a new (though equivalent) definition and extend
delay-robustness naturally to multiple events. Then we prove that
delay-robustness for a set $R_2$ (of multiple events) implies that
delay-robustness holds for any subset of $R_2$.

\begin{definition}\label{def2}
Let $R_2 \subseteq {\Sigma_2}$ be a subset of events $r$ imported from ${\bf G}_2$ by ${\bf LOC}_1$ via their corresponding channels ${\bf CH}(2,r,1)$ (i.e. $\Sigma_{ch} = R_2$), and let ${\bf SUP}_1$ be modified to ${\bf SUP}_1'$ by replacing each $r$ by its transmitted version $r'$ as before. Let
\[{\bf SUP}' := Sync({\bf SUP}_1',\{{\bf CH}(2,r,1)| r\in R_2\}, {\bf SUP}_2).\]
Then $\bf SUP$ is {\it delay-robust relative to the event subset}
$R_2$ provided $Supqc({\bf SUP}', \\Null[\{r'|r\in R_2\}])$ is
isomorphic to $\bf SUP$.
\end{definition}


Note that the property of $\bf SUP$ described in Definition~\ref{def2} is stricter than in Definition~\ref{def1}: that $\bf SUP$ is delay-robust with respect to each event $r \in R_2$ taken separately does not imply that $\bf SUP$ is delay-robust with respect to $R_2$ as a subset; however, that $\bf SUP$ is delay-robust with respect to $R_2$ does imply that $\bf SUP$ is delay-robust with respect to each separate event $r \in R_2$. The former statement will be confirmed by Example~\ref{exmp2a} and the latter by Theorem~\ref{thm2}.

\begin{figure}[!t]
\centering
    \includegraphics[scale=0.5]{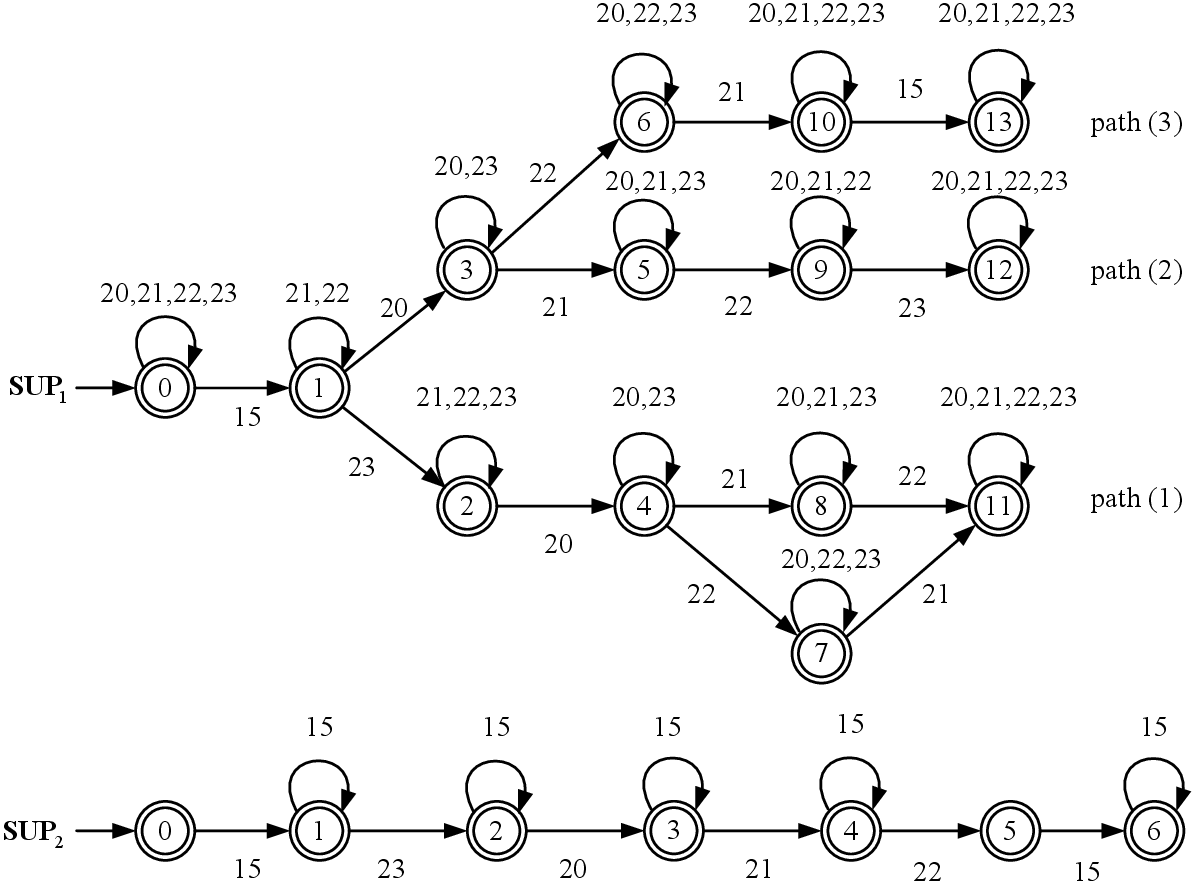}\\
\caption{Example~\ref{exmp2a}: ${{\bf SUP}_1}$ and
${{\bf SUP}_2}$} \label{exmp2a:fig1}
\end{figure}
\begin{figure}[!t]
\centering
    \includegraphics[scale=0.5]{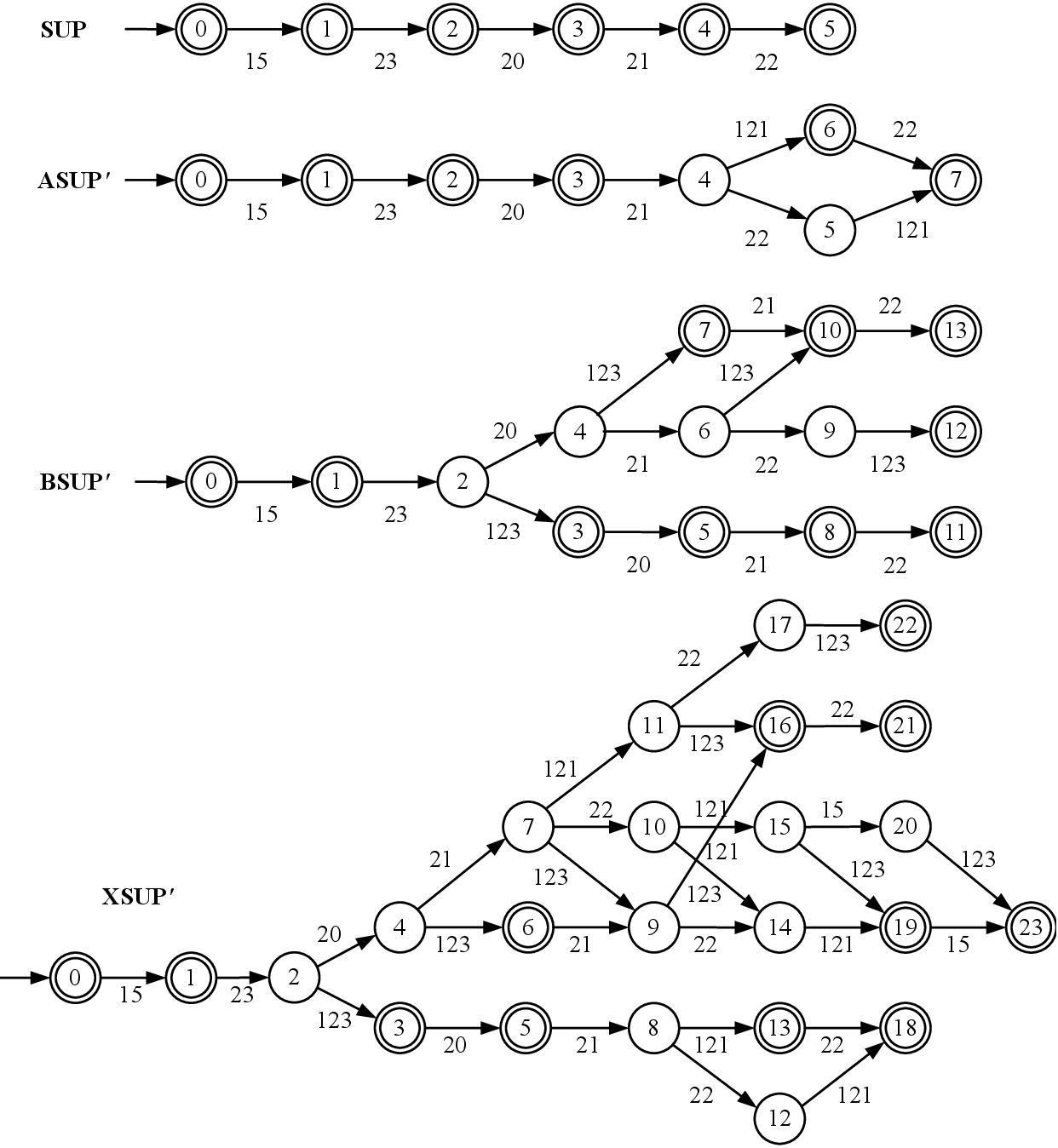}\\
  \caption{Example~\ref{exmp2a}: ${\bf SUP}$, ${\bf ASUP}'$, ${\bf BSUP}'$ and ${\bf XSUP}'$}
  \label{exmp2a:fig2}
\end{figure}

\begin{example} \label{exmp2a}
In this example $\bf SUP$ is delay-robust with respect to events 21
and 23 separately, but is not delay-robust with respect to the event
set $\{21, 23\}$. Let ${\bf SUP}_1$ and ${\bf SUP}_2$ be the
generators shown in Fig.~\ref{exmp2a:fig1}, where events 20,21,22,23
in ${\bf SUP}_2$ are exported to ${\bf SUP}_1$ and event 15 in ${\bf
SUP}_1$ is exported to ${\bf SUP}_2$. Let events 21 and 23 be
transmitted by communication channel ${\bf CH}(2,21,1)$ (with signal
event 121) and ${\bf CH}(2,23,1)$ (with signal event 123) respectively.
Let ${\bf ASUP}_1'$ (resp. ${\bf BSUP}_1'$) be obtained by replacing
$21$ (resp. 23) in ${\bf SUP}_1$ by $121$ (resp. 123) and ${\bf
XSUP}_1'$ be obtained by simultaneously replacing 21 and 23 in
${\bf SUP}_1$ by 121 and 123. Let
\begin{align*}
{\bf SUP} &= Sync({\bf SUP}_1, {\bf SUP}_2) \\
{\bf ASUP}' &= Sync({\bf ASUP}_1', {\bf CH}(2,21,1), {\bf SUP}_2) \\
{\bf BSUP}' &= Sync({\bf BSUP}_1', {\bf CH}(2,23,1), {\bf SUP}_2) \\
{\bf XSUP}' &= Sync({\bf XSUP}_1', {\bf CH}(2,21,1),  {\bf CH}(2,23,1), {\bf SUP}_2),
\end{align*}
as shown in Fig.~\ref{exmp2a:fig2}. One can verify that both
$Supqc({\bf ASUP}', Null[121])$ and $Supqc({\bf BSUP}', Null[123])$
are isomorphic to $\bf SUP$, i.e. $\bf SUP$ is delay-robust with
respect to 21 and 23 separately. However, $\bf SUP$ is not
delay-robust with respect to the event set $\{21, 23\}$. Take
\begin{align*}
s = 15.23.20.123.21.22.121.15.
\end{align*}
As in Fig.~\ref{exmp2a:fig2}, $s \in L({\bf XSUP}')$, but by
projecting out 121 and 123,
\begin{align*}
Ps = 15.23.20.21.22.15 \notin L({\bf SUP}),
\end{align*}
which implies that $PL({\bf XSUP}') \nsubseteq L({\bf SUP})$ (where
$P$ is the natural projection which projects 121 and 123 to the
empty string $\epsilon$).

Intuitively, one sees from Fig.~\ref{exmp2a:fig1} that ${\bf SUP}_1$
at its state 1 has three paths to choose from: paths (1) and (2) are
`safe', but path (3) is `dangerous' (because event 15 will occur,
which violates $\bf SUP$'s behavior). Which path ${\bf SUP}_1$ chooses
depends on the events imported from ${\bf SUP}_2$. If event 21 alone
is delayed, ${\bf SUP}_1$ can choose only path (1); if event 23 alone
is delayed, ${\bf SUP}_1$ can choose either path (1) or (2); thus
delaying 21 and 23 individually leads only to `safe' paths. If,
however, events 21 and 23 are both delayed, ${\bf SUP}_1$ can choose
any of the three paths including the `dangerous' path (3).

\end{example}


Before addressing delay-robustness for event subsets, we extend our
definition to the general case with $n$ agents ${\bf G}_j$ ($j \in N
= \{1, 2, ..., n\}$), each with local controller ${\bf LOC}_j$ which
imports channeled events $\Sigma_{ch}(i,j) \subseteq \Sigma_i$ from
${\bf G}_i$ ($i \in I_j \subset N$). For this configuration we
employ binary channels as before, one for each $r \in \Sigma_{
ch}(i,j)$. Thus an event $r \in \Sigma_i$ that is channeled to both
${\bf LOC}_j$ and ${\bf LOC}_k$ will employ separate channels ${\bf
CH}({i,r,j})$ and ${\bf CH}({i,r,k})$. Here the channels
${\bf CH}({i,r,j})$ and ${\bf CH}({i,r,k})$ are distinct (see
Fig.~\ref{fig:multichn}): we use different signal events $r_j'$ and
$r_k'$ corresponding to $r$ in ${\bf CH}({i,r,j})$ and ${\bf CH}({i,r,k})$,
respectively; in this way, the channeled event $r$ may be received by
${\bf LOC}_j$ and ${\bf LOC}_k$ in either order and with unspecified
delays. Of course $r$ might also be communicated (but with zero delay)
from ${\bf G}_i$ to other local controllers ${\bf LOC}_l$ with $l \neq j, k$.

\begin{figure}[!t]
\centering
    \includegraphics[scale=0.5]{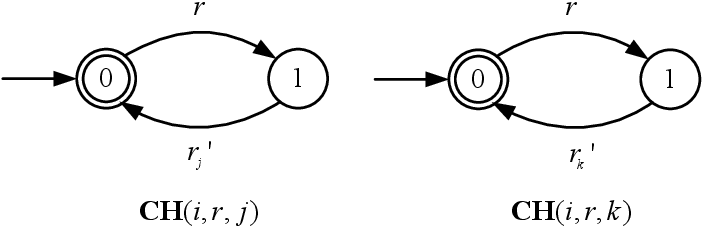}\\
\caption{\small ${\bf CH}(i,r,j)$ and ${\bf CH}(i,r,k)$, with
distinct signal events $r'_j$ and $r'_k$} \label{fig:multichn}
\end{figure}

For this architecture, Definition~\ref{def2} is generalized in the
obvious way. For each $j \in N$ we compute ${\bf SUP}_j'$ by
relabeling each event $r$ that appears in ${\bf SUP}_i$, such that
$r \in \Sigma_{ch}(i,j)$ ($i \in I_j$), by its channeled output
$r'$. Since $\Sigma_{ch}(i,j) \subseteq \Sigma_i$ and the $\Sigma_i$
are pairwise disjoint, this relabeling is unambiguous. Then we
compute
\begin{align} \label{eq:SUP'}
{\bf SUP}' = Sync({\bf SUP}_j', {\bf CH}({i,r,j}) \ |\ r \in
\Sigma_{ch}(i,j), i \in I_j, j \in N)
\end{align}
Note that if for some $j$, $I_j = \emptyset$, i.e. ${\bf LOC}_j$
imports no events from other agents ${\bf G}_i$, $i \neq j$, then
${\bf SUP}_j'={\bf SUP}_j$.

With ${\bf SUP} = Sync({\bf SUP}_j \ |\ j \in N)$, we have the
following definition.

\begin{definition}\label{def_allevents}
$\bf SUP$ is \emph{delay-robust for distributed control of $n$
agents by localization} provided the projected channeled behavior
\begin{align} \label{eq:def_allevents}
Supqc({\bf SUP}', Null\{r'| r\in \Sigma_{ch}(i,j), i \in I_j, j \in
N\})
\end{align}
is deterministic, and isomorphic with $\bf SUP$.
\end{definition}

The justification of this definition is merely a repetition of the
argument for two agents based on the conditions (\ref{e4a}),
(\ref{e4b}) and (\ref{e2c}bis). Once the obvious generalization of
${\bf SUP}'$ has been framed, as above, the basic conditions just
referenced are fully defined as well, and require no formal change.
The final result in terms of $Supqc$ is derived exactly as before.

We note that to verify delay-robustness in
Definition~\ref{def_allevents} we need to compute ${\bf SUP}'$ as in
(\ref{eq:SUP'}). The computation may be expensive when there is a
large number of communication channels. Nevertheless ${\bf SUP}'$ is
implemented in a purely distributed fashion: distributed supervisors
and communication channels. We shall investigate the computational
issue of ${\bf SUP}'$ in our future work, one promising approach
being to use \emph{State Tree Structures} \cite{MaWon:05}. We also
note in passing that all the above results can be extended to
decentralized controllers; for details see Appendix~\ref{appB}.

In the foregoing notation now suppose that $\bf SUP$ is known to be
delay-robust for a set of binary channels ${\bf CH}({i,r,j})$ with
$i \in I_j$, $j \in N$, and $r$ in some subset $\Sigma_{ch}(i,j)
\subseteq \Sigma_i$. We shall prove that $\bf SUP$ \emph{remains
delay-robust when any one of these channels is replaced by the ideal
channel with zero transmission delay}.  As a corollary,
delay-robustness is preserved if the given set $\Sigma_{ch}(i,j)$
of channeled events from ${\bf G}_i$ to ${\bf LOC}_j$ is
replaced by any subset. Focussing attention on ${\bf
SUP}_1=Sync({\bf G}_1,{\bf LOC}_1)$, consider its environment
$E=\{{\bf SUP}_2,\ldots,{\bf SUP}_N\}$ with ${\bf SUP}_E :=
Sync\{{\bf SUP}_i \ |\ i=2,\ldots,N\}$. We assume that $E$ is
augmented to a channeled version $E'$ (say) having internal channels
${\bf CH}(i,r_{ij},j)$ $(i,j=2,...,N, i \neq j, r_{ij} \in
\Sigma_{ch}(i,j))$, together with outgoing external channels ${\bf
CH}(j,r_{j1},1)$ to ${\bf LOC}_1$ and incoming external channels
${\bf CH}(1,r_{1i},i)$ from ${\bf G}_1$. Denote the totality of
$E$'s internal channels, along with those from ${\bf G}_1$, by ${\bf
CH}_E$. Write ${\bf SUP}_{E'} := Sync({\bf SUP}'_2,...,{\bf
SUP}'_N,{\bf CH}_E)$ where ${\bf SUP}'_j$ is ${\bf SUP}_j$ with any
event $r \in \Sigma_{ch}(i,j)$ replaced by $r'$ $(i=1,...,N;
j=2,...,N; i \neq j)$ as prescribed before. For the alphabet of
${\bf SUP}_{E'}$ we have
\begin{align*}
\Sigma_{E'} = \cup \{ \Sigma_i \ |\ i=2,...,N\} \ \cup\ \{r' \ |\ r
\in \Sigma_{ch}(i,j); i=1,...,N, j=2,...,N, i \neq j\}.
\end{align*}
Similarly let ${\bf SUP}'_1$ denote ${\bf SUP}_1$ with channeled
events $r_{j1} \in \Sigma_{ch}(j,1)$ $(j=2,...,N)$ replaced by
$r'_{j1}$, and let $\Sigma'_1$ denote the corresponding alphabet. By
assumption the alphabets $\Sigma_i$ $(i=1,...,N)$ are pairwise
disjoint, hence the $\Sigma_{ch}(j,1)$ $(j=2,...,N)$ together with
$\Sigma_1$ are pairwise disjoint.  Write
\begin{align*}
\Sigma_{ch}(E,1) = \cup \{\Sigma_{ch}(j,1) \ |\ j=2,...,N\}.
\end{align*}
For clarity assume $\Sigma_{ch}(E,1)=\{\alpha, \beta\}$; the
extension to more than two events will be evident.  Thus $\alpha$,
$\beta$ are the channeled events imported to ${\bf LOC}_1$ from its
environment ${\bf SUP}_E$ (actually ${\bf SUP}_{E'}$), and appear in
${\bf SUP}'_1$ as $\alpha'$, $\beta'$. We can therefore write ${\bf
SUP}'$ in (\ref{eq:def_allevents}) in more detail as
\begin{align*}
{\bf SUP}' = Sync ({\bf SUP}'_1, {\bf CH}(E,\alpha,1),{\bf
CH}(E,\beta,1),{\bf SUP}_{E'}).
\end{align*}
Notice that $\alpha$, $\beta$ belong to $\Sigma_E := \Sigma_2 \cup
\cdots \cup \Sigma_N$ but not $\Sigma_1$, whereas $\alpha'$,
$\beta'$ appear in ${\bf SUP}'_1$ and the two channels but not in
${\bf SUP}_{E'}$.

Now denote by ${\bf SUP}''$ the structure ${\bf SUP}'$ but with the
channel ${\bf CH}(E,\alpha,1)$ replaced by one with zero delay (and
so eliminated from the channel formalism). Thus
\begin{align*}
{\bf SUP}'' = Sync ({\bf SUP}''_1, {\bf CH}(E,\beta,1),{\bf
SUP}_{E'})
\end{align*}
where ${\bf SUP}''_1$ is ${\bf SUP}_1$ with $\beta$ replaced by
$\beta'$ (but $\alpha$ left unchanged). We shall prove the following
result.

\begin{theorem} \label{thm2}
If $\bf SUP$ is delay-robust with respect to the channel structure
of ${\bf SUP}'$, then it remains so with respect to that of ${\bf
SUP}''$.
\end{theorem}

The assertion is almost obvious from the intuition that the
statement for ${\bf SUP}''$ should be derivable by ``taking the
limit'' at which ${\bf CH}(E,\alpha,1)$ operates with zero delay,
namely by replacing the communication event $\alpha$, when
unchanneled, with the zero-delay channeled version $\alpha.\alpha'$,
and finally projecting out $\alpha'$. A proof is given in
Appendix~\ref{appC}.



\subsection{Blocking of Uncontrollable Events} \label{sec:3.3}


The foregoing discussion of delay robustness covers channeled
events in general, regardless of their control status, and
is adequate if all channeled events happen to be controllable.
In the case of uncontrollable channeled events, however,
we must additionally examine whether channel delay violates
the conventional modeling assumption that uncontrollable events
may occur spontaneously at states where they are enabled and
should not be subject to external disablement. 

In our simplified model the transmission of $r$ from ${\bf G}_2$ to
${\bf LOC}_1$ is completed (by event $r'$) with indefinite (unbounded)
delay. A constraint imposed on ${\bf SUP}'$ by the channel ${\bf CH}(2,r,1)$
is that $r$ cannot occur again until $r'$ has reset ${\bf CH}(2,r,1)$ and
the communication cycle is ready to repeat. If $r$ is controllable
its re-occurrence can be disabled and hence delayed until after the
occurrence of $r'$ corresponding to the previous occurrence of $r$.
If, however, $r$ is uncontrollable, then once it is re-enabled (by
entrance of ${\bf SUP}_2$ to a state where $r$ is defined) its
re-occurrence cannot be externally delayed, according to the usual
modeling assumption on uncontrollable events. In this sense the
introduction of ${\bf CH}(2,r,1)$ could conceivably conflict with the
intention of the original DES model. To address this issue we
examine whether or not communication delay of an
uncontrollable event might violate a modeling assumption.




\begin{example} \label{exmp3}
For illustration, let ${\bf SUP}_1$ and ${\bf SUP}_2$ be the
generators shown in Fig.~\ref{fig12}. Assume event 20 in ${\bf
SUP}_2$ is exported to ${\bf SUP}_1$, i.e., $r = 20$ and $r' = 120$;
${\bf SUP}_1'$ is obtained by replacing $20$ in ${\bf SUP}_1$ by
$120$.  As shown in Fig.~\ref{fig13}, ${\bf SUP}' = Sync({\bf
SUP}_1', {\bf CH}(2,20,1), {\bf SUP}_2)$ is easily verified to be
delay-robust with respect to event $20$. Define ${\bf NSUP} =
Sync({\bf SUP}_1', {\bf SUP}_2)$. Let $s = 20$; then $s.20 \in
L({\bf NSUP})$, but $s.20 \notin L({\bf SUP}')$. Since ${\bf SUP}' =
Sync({\bf NSUP}, {\bf CH}(2,20,1))$, event 20 is blocked by ${\bf
CH}(2,20,1)$.
\end{example}
\begin{figure}[!t]
\centering
    \includegraphics[scale=0.55]{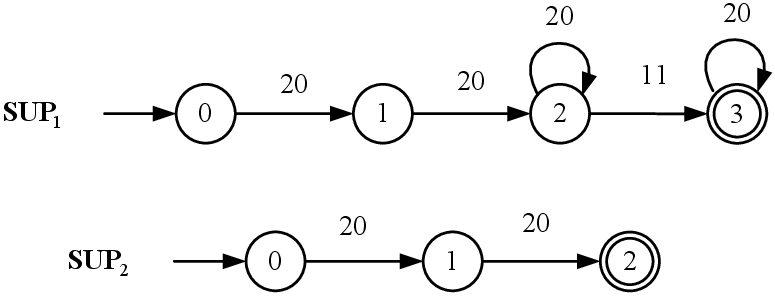}\\
  \caption{Example~\ref{exmp3}: ${{\bf SUP}_1}$ and ${{\bf SUP}_2}$}
  \label{fig12}
\end{figure}

\begin{figure}[!t]
\centering
    \includegraphics[scale=0.55]{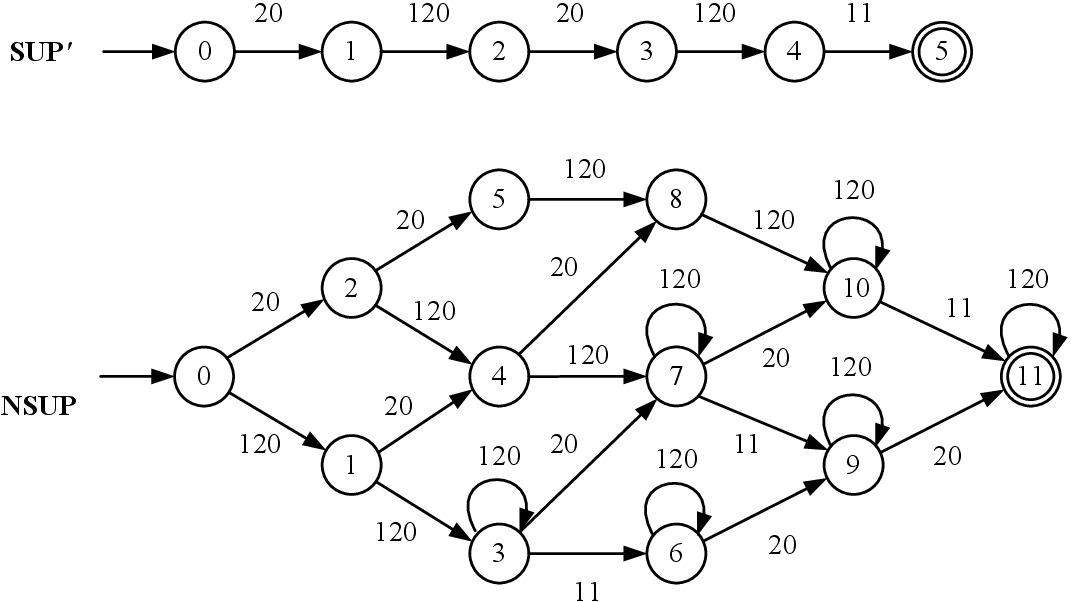}\\
  \caption{Example~\ref{exmp3}: ${\bf SUP}'$ and ${\bf NSUP}$}
  \label{fig13}
\end{figure}

This example shows a case where the reoccurrence of an
uncontrollable event is `blocked' by its channel, which demonstrates
that communication delay of an uncontrollable event really violates
the modeling assumption that uncontrollable events cannot be
disabled by any external agent. Now let
\begin{equation} \label{e6a}
{\bf NSUP} = Sync({\bf SUP}_1', {\bf SUP}_2);
\end{equation}
then according to (\ref{e1a})
\begin{equation} \label{e6b}
{\bf SUP}' = Sync({\bf NSUP}, {\bf CH}(2,r,1)).
\end{equation}
As before, write $\Sigma' = \Sigma \cup \{r'\}$ for the
alphabet of ${\bf SUP}'$, let $P: \Sigma'^* \rightarrow \Sigma^*$
be the natural projection of $\Sigma'^*$ to $\Sigma^*$, and
define the new natural projection
$P_r: \Sigma'^* \rightarrow \{r, r'\}^*$.
Now, for given $\bf NSUP$ and ${\bf SUP}'$ as in (\ref{e6a}) and
(\ref{e6b}), and $r \in \Sigma_u$, if there exists $s\in L({\bf
SUP}')$ such that $sr \in L({\bf NSUP})$, but $sr \notin L({\bf
SUP}')$, then we say that $r$ is {\it blocked} by ${\bf
CH}(2,r,1)$.


To check whether or not $r$ is blocked by ${\bf CH}(2,r,1)$, we
check if $P_r^{-1}L({\bf CH}(2,r,1))$ is ${\bf NSUP}$-controllable
with respect to event $r$, i.e.
\[P_r^{-1}L({\bf CH}(2,r,1))r \cap L({\bf NSUP}) \subseteq P_r^{-1}L({\bf CH}(2,r,1)).\]
For this, we employ the standard algorithm that checks
controllability\cite{Wonham:2011a}; the algorithm has complexity
$O(mn)$ where $m$ and $n$ represent the state numbers of ${\bf
CH}(2,r,1)$ and ${\bf NSUP}$, respectively.\footnote{For the case
described in Section~\ref{sec:3.2} of transmitting multiple events
by separate channels, we use the same method to check if each event
$r$ is blocked. Specifically, we check if $P_r^{-1}L({\bf
CH}(i,r,j))$ is $\bf NSUP$-controllable with respect to $r$, where
$\bf NSUP$ denotes the behavior of the system excluding ${\bf
CH}(i,r,j)$.}

To summarize, for an uncontrollable event $r$, if $\bf SUP$ is
delay-robust (by Theorem~\ref{thm1}) and $r$ will not be blocked by
${\bf CH}(2,r,1)$ (by controllability checking algorithm), then
$\bf SUP$ is said to be `unbounded' delay-robust with respect to
$r$. Otherwise, there exists $s \in L({\bf SUP}')$ such that $sr \in
L({\bf NSUP})$, but $sr \notin L({\bf SUP}')$. Thus $r$ is blocked
by the channel, which could violate the modeling assumption that an
uncontrollable event should never be prohibited or delayed by an
external agent. However, if the occurrence of $r'$ is executed by
${\bf LOC}_1$ before the next occurrence of $r$, the controllers may
still achieve global optimal nonblocking supervision. In this case,
we say that $\bf SUP$ is `bounded' delay-robust with respect to
$r$.\footnote{One way to determine a delay bound in terms of number
of event occurrences is to find the shortest path between two
consecutive occurrences of event $r$ in $\bf SUP$. A more detailed
study of this issue is left for future research.}

We illustrate the foregoing results by an example adapted from
\cite{Wonham:2011a}.



\section{Example - WORKCELL} \label{sec:4}

\subsection{Model Description and Controller Design} \label{sec:4.1}


$\bf WORKCELL$ consists of $\bf ROBOT$, $\bf LATHE$ and $\bf FEEDER$, with three buffers, $\bf INBUF$, $\bf LBUF$ and $\bf SBBUF$, connected as in Fig.~\ref{fig4}. Labeled arrows denote synchronization on shared transitions (events) in the corresponding component DES.
\begin{figure}[!t]
\centering
\includegraphics[scale=0.5]{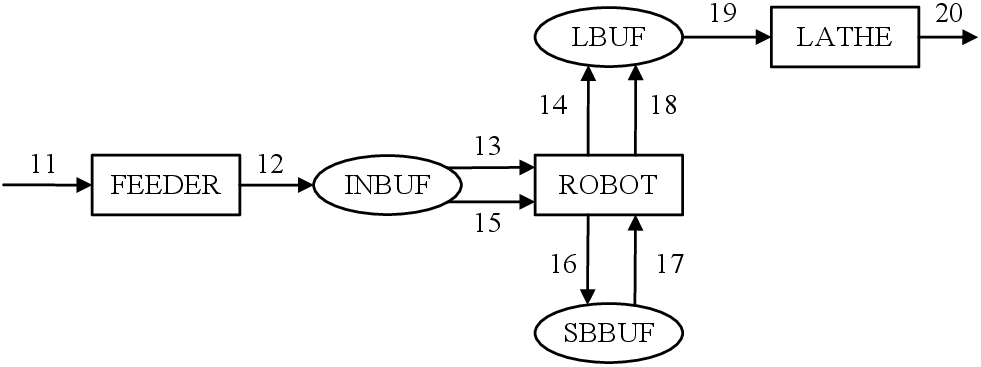}
\caption{{\bf WORKCELL}}
\label{fig4}
\end{figure}

$\bf WORKCELL$ operates as follows: $\bf FEEDER$ acquires a new part
from an infinite source (event 11) then stores it (event 12) in a
2-slot buffer $\bf INBUF$. $\bf ROBOT$ takes a new part from $\bf
INBUF$ (event 13) and stores it (event 14) in a 1-slot buffer $\bf
LBUF$; if $\bf LBUF$ is already full, $\bf ROBOT$ may instead take a
new part from $\bf INBUF$ (event 15) and store it (event 16) in a
1-slot `stand-by' buffer $\bf SBBUF$. If $\bf LBUF$ is empty and
there's already a part in $\bf SBBUF$, $\bf ROBOT$ first unloads the
part in $\bf SBBUF$ (event 17) and loads it in $\bf LBUF$ (event
18). If $\bf LATHE$ is idle and there exists a part in $\bf LBUF$, $\bf
LATHE$ takes that part and starts working on it (event 19), and when
finished exports it and returns to idle (event 20). Event labels
accord with\cite{Wonham:2011b}: odd-(resp. even-) numbered events are
controllable (resp. uncontrollable). The physical interpretations of events are
displayed in Table~\ref{tab1}.
\begin{table}
\caption{Physical interpretation of events}
\label{tab1}
\begin{center}
\scalebox{0.95}{
\begin{tabular}{|c||l|}
\hline
Event label & Physical interpretation \\
\hline
11 & {\bf FEEDER} imports new part from infinite source \\
\hline
12 & {\bf FEEDER} loads new part in {\bf INBUF} \\
\hline
13 & {\bf ROBOT} takes part from {\bf INBUF} for loading into {\bf LBUF} \\
\hline
14 & {\bf ROBOT} loads part from {\bf INBUF} into {\bf LBUF} \\
\hline
15 & {\bf ROBOT} takes part from {\bf INBUF} for loading into {\bf SBBUF} \\
\hline
16 & {\bf ROBOT} loads part from {\bf INBUF} into {\bf SBBUF} \\
\hline
17 & {\bf ROBOT} takes part from {\bf SBBUF} for loading into {\bf LBUF} \\
\hline
18 & {\bf ROBOT} loads part from {\bf SBBUF} into {\bf LBUF} \\
\hline
19 & {\bf LATHE} loads part from {\bf LBUF} and starts working \\
\hline
20 & {\bf LATHE} exports finished part and returns to idle\\
\hline
\end{tabular}
}
\end{center}
\end{table}

The specifications to be enforced are: 1) ${\bf SPEC}_1$ says that a
buffer must not overflow or underflow; 2) ${\bf SPEC}_2$ says that
$\bf ROBOT$ can load $\bf SBBUF$ (event sequence 15.16) only when
$\bf LBUF$ is already full; 3) ${\bf SPEC}_3$ says that $\bf ROBOT$
can load $\bf LBUF$ directly from $\bf INBUF$ (event sequence 13.14)
only when $\bf SBBUF$ is empty; otherwise it must load from $\bf
SBBUF$ (event sequence 17.18). The DES models of plant components
and specifications are shown in Figs.~\ref{fig6} and \ref{fig7}.
\begin{figure}[!t]
\centering
    \includegraphics[scale=0.5]{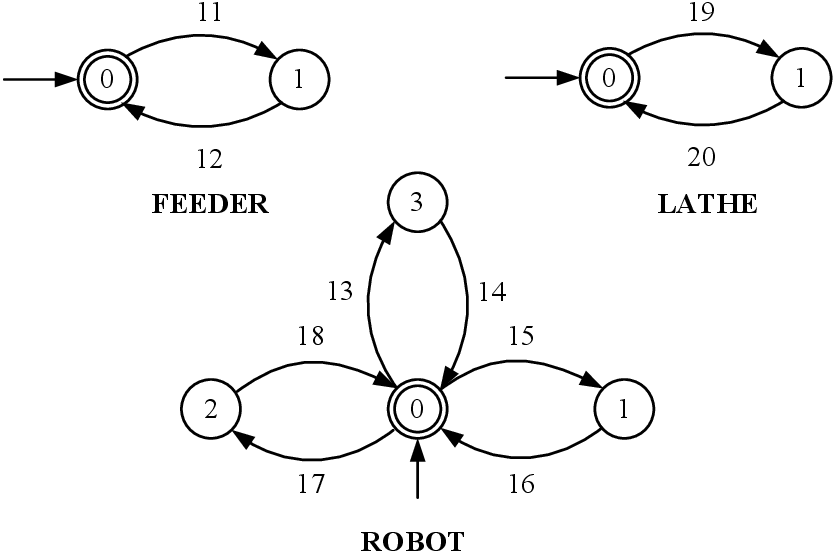}\\
\caption{Plant models to be controlled}
\label{fig6}
\end{figure}
\newline

\begin{figure}[!t]
\centering
    \includegraphics[scale=0.5]{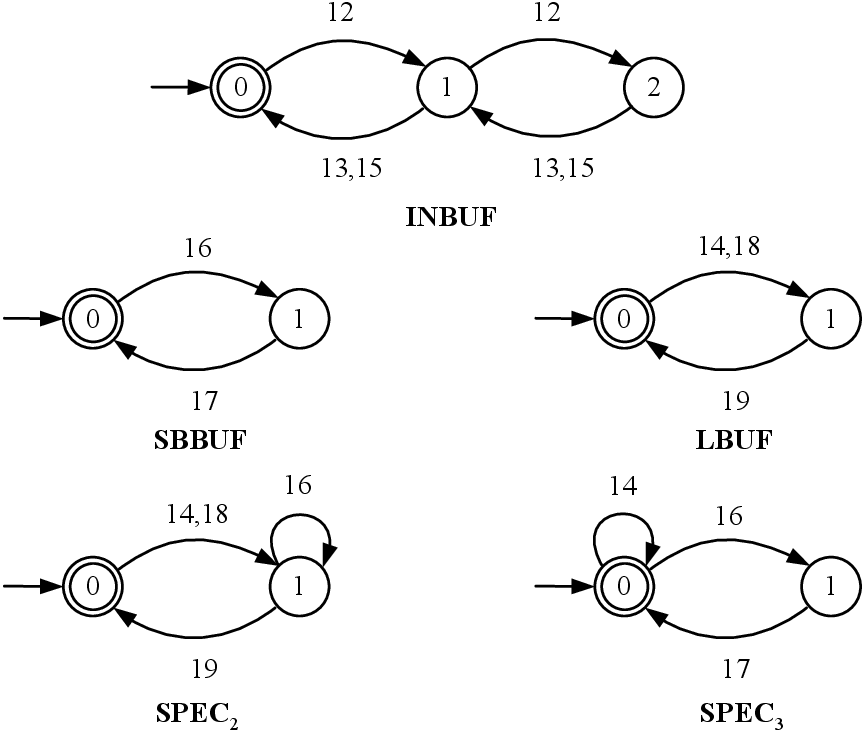}\\
\caption{Model of Specifications}
\label{fig7}
\end{figure}

We first compute the monolithic supervisor by a standard method
(e.g. \cite{Wonham:2011a,Wonham:2011b}). The behavior of $\bf
WORKCELL$ is the synchronous product of $\bf FEEDER$, $\bf ROBOT$,
and $\bf LATHE$. As ${\bf SPEC}_1$ is automatically incorporated in
the buffer models, the total specification $\bf SPEC$ is the
synchronous product of $\bf INBUF$, $\bf LBUF$, $\bf SBBUF$, ${\bf
SPEC}_2$, and ${\bf SPEC}_3$. The monolithic supervisor is ${\bf
SUPER} = Supcon({\bf WORKCELL}, {\bf SPEC})$ with (state,
transition) count (70, 153).

Next by use of procedure {\it Localize}\cite{Wonham:2011a,Wonham:2011b}, we compute the localization of $\bf SUPER$ (in the sense of \cite{CaiWonham:2010a,CaiWonham:2010b}) to each of the three $\bf WORKCELL$ agents, to obtain local controllers $\bf FEEDERLOC$, $\bf ROBOTLOC$ and $\bf LATHELOC$, as shown in Fig.~\ref{fig8}. The local controlled behaviors are
\begin{align*}
{\bf FEEDERSUP} &= Sync({\bf FEEDER}, {\bf FEEDERLOC}),\\
{\bf ROBOTSUP} &= Sync({\bf ROBOT}, {\bf ROBOTLOC}),  \\
{\bf LATHESUP} &= Sync({\bf LATHE}, {\bf LATHELOC}).
\end{align*}
From the transition structures shown in Fig.~\ref{fig8}, we see that $\bf FEEDERLOC$ ($\bf FEEDERSUP$) must import events 13, 14, 15, 16, 17 and 18 from $\bf ROBOT$, and 19 from $\bf LATHE$; $\bf ROBOTLOC$ ($\bf ROBOTSUP$) must import events 12 from $\bf FEEDER$, and 19 from $\bf LATHE$; and $\bf LATHELOC$ ($\bf LATHESUP$) must import events 11 and 12 from $\bf FEEDER$, and 13, 14, 15, 16, 17 and 18 from $\bf ROBOT$.
\begin{figure}[!t]
\centering
    \begin{minipage}{0.48\linewidth}
        \centering
        \begin{overpic}[scale = 0.5]{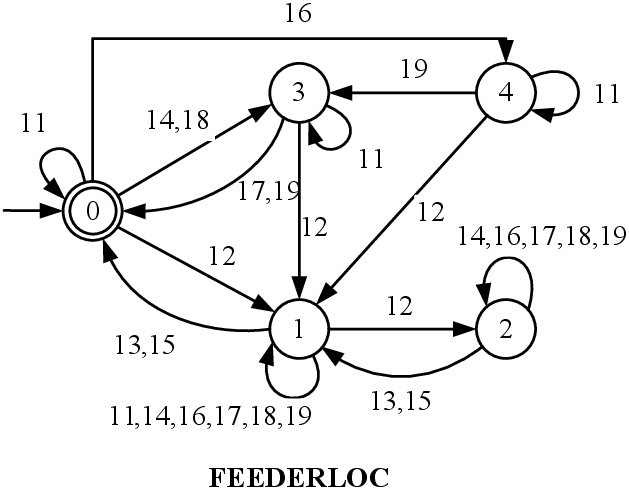}
        \end{overpic}
    \end{minipage}
    \hfill
    \begin{minipage}{0.48\linewidth}
        \centering
        \begin{overpic}[scale = 0.5]{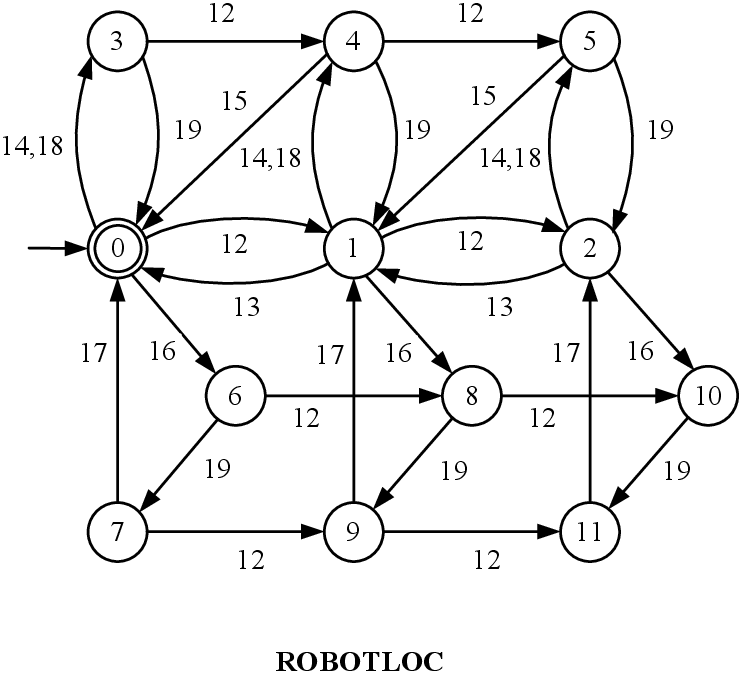}
        \end{overpic}
    \end{minipage}\\
    \vspace{2.5em}
    \begin{minipage}{0.45\linewidth}
        \begin{overpic}[scale = 0.5]{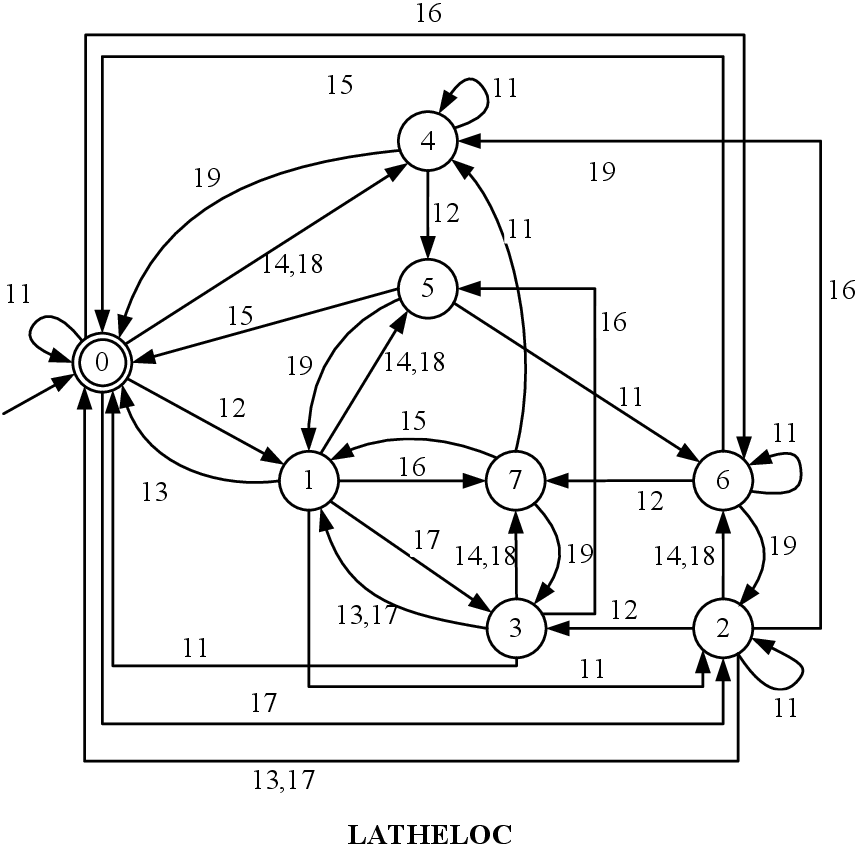}
        \end{overpic}\\
    \end{minipage}\\
\caption{Local Controller for each component. According to Remark 1, for every state $x$ of each controller, and each communication event $\sigma$ imported from some other component, if $\sigma$ is not defined at x, we add a $\sigma$-selfloop. Let $*(x)$ be the set of selfloops to be adjoined at state $x$. In ${\bf FEEDERLOC}$, $*(0) = \{13,15,17,19\}$, $*(3) = \{13, 14,15,16,18\}$, $*(4) = \{13,14,15,16,17,18\}$; in ${\bf ROBOTLOC}$, $*(0) = \{19\}$, $*(1) = \{19\}$, $*(2) = \{12,19\}$, $*(5) = \{12\}$, $*(7) = \{19\}$,$*(9) = \{19\}$,$*(10) = \{12\}$,$*(11) = \{19\}$; in ${\bf LATHELOC}$, $*(0) = \{13,15\}$, $*(1) = \{12,15\}$, $*(2) = \{15\}$, $*(3) = \{12,15\}$, $*(4) = \{13,14,15,16,17,18\}$,$*(5) = \{12,13,14,16,17,18\}$,$*(6) = \{13,14,16,17,18\}$,$*(7) = \{12,13,14,16,17,18\}$.}
\label{fig8}
\end{figure}

\subsection{Illustrative Cases} \label{sec:4.2}


Based on the computed local controllers, we illustrate our new verification tools with the following cases.

\begin{case} \label{case1}
-- Event 13

Taking $\bf FEEDERLOC$ for example, build a  channel ${\bf CH}(R,13,F)$, as
shown in Fig.~\ref{fig9}, using a new event label 113 to represent
the corresponding channel output; use 113 to replace 13 in $\bf
FEEDERSUP$ to obtain ${\bf FEEDERSUP}'$, over the alphabet
\{11,12,113,14,15,16,17,18,19\}.

\begin{figure}[!t]
\centering
    \includegraphics[scale=0.5]{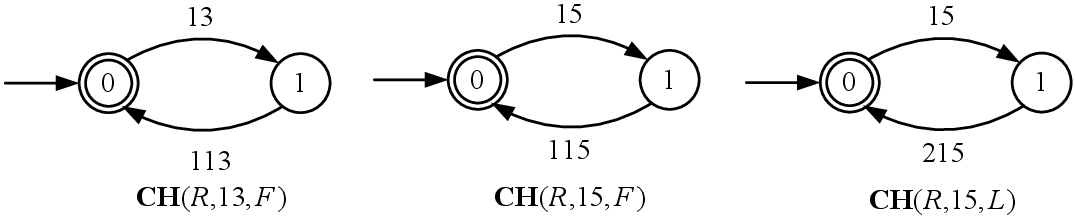}\\
\caption{\small ${\bf CH}(R,13,F)$, ${\bf CH}(R,15,F)$, and ${\bf
CH}(R,15,L)$} \label{fig9}
\end{figure}

Now compute the channeled behavior ${\bf SUPER}'$ according to
\begin{align*}
{\bf SUPER}' = Sync({\bf FEEDERSUP}', {\bf CH}(R,13,F), {\bf ROBOTSUP}, {\bf LATHESUP})
\end{align*}
\noindent over the augmented alphabet \{11, ..., 20, 113\} and with
(state, transition) count (124, 302). Next, to check
delay-robustness we project ${\bf SUPER}'$ modulo supremal
quasi-congruence with nulled event 113, to get, say,
\begin{align*}
{\bf QCSUPER'} ~:=~ Supqc(&{\bf SUPER}', Null[113])\\
                        &(deterministic, with\ size\ (70, 153))
\end{align*}

Finally we verify that ${\bf QCSUPER}'$ is isomorphic to $\bf
SUPER$, and conclude that $\bf SUPER$ is delay-robust with respect
to the channeled communication of event 13 from $\bf ROBOT$ to $\bf
FEEDERLOC$. As a physical interpretation, consider the case where
events 11, 12, 11, 12, 13 have occurred sequentially (i.e. there
exist two parts in $\bf INBUF$ and $\bf ROBOT$ has taken a part from
$\bf INBUF$) and ${\bf FEEDERSUP}'$ has not executed the occurrence
113 of event 13. On the one hand, if ${\bf FEEDERSUP}'$ executes event
113 (i.e. it acknowledges the occurrence of event 13), it will
enable event 11 legally (according to $\bf SUPER$). On the other
hand, if ${\bf FEEDERSUP}'$ does not execute event 113, then $\bf
ROBOT$ will load the part into $\bf LBUF$ and take another part from
$\bf INBUF$ (execute event 15). So ${\bf FEEDERSUP}'$ can enable event
11 again, which is also legal according to $\bf SUPER$. Hence, in
this case, the channeled system ${\bf SUPER}'$ can run `correctly'(no
extra behavior violates the specification) and can `complete' the
given task (with the help of $\bf SBBUF$), i.e. the communication
delay of event 13 is tolerable with respect to $\bf SUPER$.

By the same method, one can verify that $\bf SUPER$ is delay-robust
with respect
event 15 \emph{provided} it is channeled only to ${\bf
FEEDERLOC}$; it must be communicated to ${\bf LATHELOC}$ without
delay. To verify this, we have two separate channels, ${\bf CH}({R,15,F})$ and ${\bf CH}(R,15,L)$, with distinct signal events 115 and 215 (see Fig.~\ref{fig9}). Taking the two channels separately, by Definition~\ref{def1} and the same method as above for event 13, we verify that $\bf SUPER$ is delay-robust when 15 is communicated to $\bf FEEDERLOC$ by ${\bf CH}(R,15,F)$, but delay-critical to $\bf LATHELOC$ by ${\bf CH}(R,15,L)$. Moreover, by Definition~\ref{def_allevents} and the procedure in Sect.~\ref{sec:3.2}, we verify that $\bf SUPER$ is delay-critical when 15 is communicated to both $\bf FEEDERLOC$ and $\bf LATHELOC$.
\end{case}

\begin{case} \label{case2}
-- Events 13 and 15

This case shows that $\bf SUPER$ is delay-robust relative to the
event set \{13, 15\}, with 13 and 15 both channeled to ${\bf FEEDERLOC}$.

Consider the channel ${\bf CH}(R,15,F)$ displayed in Fig.~\ref{fig9},
using the signal event 115 to represent the corresponding channel
output. Use labels 113, 115 to replace 13, 15 in $\bf FEEDERSUP$ to
obtain ${\bf FEEDERSUP}'$, over the alphabet \{11,12,113, 14, 115,16,17,18,19\}.

We  compute the channeled behavior ${\bf SUPER}'$ according to
\begin{align*}
{\bf SUPER}' = Sync({\bf FEEDERSUP}', &{{\bf CH}(R,13,F)}, {{\bf CH}(R,15,F)}, \\
&{\bf ROBOTSUP}, {\bf LATHESUP}),
\end{align*}
\noindent over the augmented alphabet \{11, ..., 20, 113, 115\} and
with (state, transition) count (180, 470). Next, to check
delay-robustness we project ${\bf SUPER}'$ modulo supremal
quasi-congruence with nulled events 113, 115, to get
\begin{align*}
{\bf QCSUPER}' := Supqc(&{\bf SUPER}', Null[113,115])\\
                    &(deterministic, with\ size\ (70, 153))
\end{align*}

Finally ${\bf QCSUPER}'$ turns out to be isomorphic to $\bf SUPER$, and we conclude that $\bf SUPER$ is delay-robust with respect to the channeled communication of events 13, 15 from $\bf ROBOT$ to $\bf FEEDERLOC$. Briefly, the reason is that ${\bf FEEDERSUP}'$ will enable event 11 after it executes event 113 or 115, and $\bf ROBOT$ will remain idle if no more parts are loaded into the system (i.e. event 11 cannot occur again).

\end{case}

\begin{case} \label{case3}
-- Event 19

Event 19 channeled to $\bf ROBOTLOC$ is shown, by computation, or directly by
Definition~\ref{def1}, to be delay-critical with respect to $\bf
SUPER$. By tracking the working process, we show that the indefinite
communication delay of event 19 may result in violation of ${\bf
SPEC}_2$. Consider the following case: events 11,12,11,12,13,14,19
have occurred sequentially, i.e. there exists one part in $\bf
INBUF$, $\bf ROBOT$ has loaded a part in $\bf LBUF$ and $\bf LATHE$
has taken the part from $\bf LBUF$ (i.e. $\bf LBUF$ is now empty).
Since the transmission of event 19 is delayed unboundedly, if $\bf
ROBOT$ doesn't `know' that $\bf LATHE$ has taken the part from $\bf
LBUF$, it may take a new part from $\bf INBUF$ (event 15) and load
it into $\bf SBBUF$ (event 16) according to $\bf ROBOTSUP'$, i.e.
the event sequence 11.12.11.12.13.14.19.15.16 occurs in $\bf
WORKCELL$ with communication delay, violating ${\bf SPEC}_2$. Hence
event 19 is delay-critical.
\end{case}

\begin{case} \label{case4}
-- Event 12

This case shows that although the occurrence of (uncontrollable)
event 12 (channelled to $\bf ROBOTLOC$) may be blocked by its channel ${\bf CH}(F,12,R)$, as shown
in Fig.~\ref{fig14}, this will not violate the specifications.
\begin{figure}[!t]
\centering
    \includegraphics[scale=0.5]{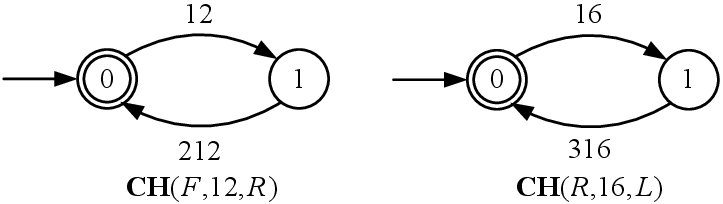}\\
\caption{${\bf CH}(F,12,R)$ and ${\bf CH}(R,16,L)$}
\label{fig14}
\end{figure}
According to Sect~\ref{sec:3.3}, we check whether $L({\bf CH}(F,12,R))$
is controllable with respect to
\begin{eqnarray*}
{\bf NSUPER} = Sync({\bf FEEDERSUP}, {\bf ROBOTSUP}', {\bf LATHESUP}).
\end{eqnarray*}
In \cite{Wonham:2011b}, we use ${Condat}$, which tabulates the set of events
disabled in ${\bf CH}(F,12,R)$ with respect to $\bf NSUPER$, to implement
the verification of the controllability for $L({\bf CH}(F,12,R))$.\footnote{
Here the alphabet of ${\bf CH}(F,12,R)$ is $\{12,212\}$; before calling ${Condat}$,
one should add the selfloop with events in $\bf NSUPER$ but not in $\{12,212\}$
at each state of ${\bf CH}(F,12,R)$.}

By using ${Condat}$, it turns out that event 12 is disabled
at state 1 of \\$L({\bf CH}(F,12,R))$.
Physically, suppose 11, 12 and 11 have occurred
sequentially, i.e., $\bf FEEDER$ has stored a part in $\bf INBUF$
and taken another part (event 11). After that, $\bf FEEDER$ may
store the part in $\bf INBUF$ (event 12, which is uncontrollable).
If $\bf ROBOTSUP$ does not
acknowledge the first occurrence of 12, then ${\bf CH}(F,12,R)$ is at state
1, and thus cannot transmit the next occurrence of 12. So, in the
channeled system ${\bf SUPER}'$, event 12 is blocked by ${\bf CH}(F,12,R)$.
If transmission of the first $12$ is completed (i.e. event 212
occurs) before the second occurrence of event 12, then event 12 will
not be blocked. In $\bf SUPER$, only event 11 occurs between two
occurrences of event 12; thus we say that $\bf SUPER$ is
`1-bound'-delay-robust with respect to event $12$.
\end{case}

\begin{case} \label{case5}
-- Event 16

This case shows that the occurrence of uncontrollable event 16 (channeled to $\bf LATHELOC$) will not be blocked by its channel ${\bf CH}(R,16,L)$, shown in Fig.~\ref{fig14}.

Applying procedure ${Condat}$ in \cite{Wonham:2011a} to
${\bf CH}(R,16,L)$, we see that 16 will not be disabled; we
conclude that event 16 will not be blocked by ${\bf CH}(R,16,L)$, and
$\bf SUPER$ is unbounded-delay-robust with respect to $16$. To
illustrate the conclusion, we consider the following case: there
exist two parts in $\bf INBUF$ and one part in $\bf LBUF$ (event
sequence 11.12.11.12.13.14.11.12); then $\bf ROBOT$ takes a part
from $\bf INBUF$ (event 15) and places it in $\bf SBBUF$ (event 16).
In Fig.~\ref{fig8}, $\bf FEEDERLOC$ is at state 2 and is waiting for
the occurrence of event 13 or 15 ($\bf ROBOT$ takes a part from $\bf
INBUF$), and enables event 11; $\bf ROBOTLOC$ is at state 8 and is
waiting for the occurrence of 19 ($\bf LATHE$ takes a part from $\bf
LBUF$) or the occurrence of event 12; and $\bf LATHELOC$ is at state
1 and is waiting for the occurrence of event 19. Now, the occurrence
of event 19 (which is enabled by $\bf LATHELOC$) will lead the
controlled plant to continue to operate. Even though $\bf LATHELOC$
does not receive the occurrence of 16, the system does not block.
Hence in this case the occurrence of event 16 is not blocked by its
channel ${\bf CH}(R,16,L)$.
\end{case}

\begin{case} \label{case6}
-- All communication events

When all communication events are subject to delay through channels
(i.e. $\Sigma_{ch} = \Sigma_{com}$), it can be verified that
delay-robustness of $\bf SUPER$ in the strong sense of
Definition~\ref{def_allevents} fails, i.e. $\bf SUPER$ fails to be
delay-robust for distributed control by localization. In fact when all the channeled events
except 19 (channeled to $\bf ROBOTLOC$) are received without delay,
Case~\ref{case6} is reduced to Case~\ref{case3}; so $\bf SUPER$
cannot be delay-robust with respect to the set of all communication
events, as asserted by Theorem~2 in Sect.~\ref{sec:3}.
\end{case}

\section{Conclusions and Future Work}

In this paper we have studied distributed control obtained by
supervisor localization on the relaxed assumption (compared to
previous literature\cite{CaiWonham:2010a,CaiWonham:2010b}) that
inter-agent communication of selected `communication events'
(channeled events) may be subject to unknown time delays. For this
distributed architecture we have identified a property of
`delay-robustness' which guarantees that the logical properties of
our delay-free distributed control (i.e. the original DES
specifications) continue to be enforced in the presence of delay,
albeit with possibly degraded temporal behavior. We have shown that
delay-robustness can be effectively tested with polynomial
complexity, and that such tests serve to distinguish between events
that are delay-critical and those that are not. The case that an
uncontrollable channeled event may be blocked by its communication
channel is identified by the algorithm for checking controllability.
A simple workcell exemplifies the approach, showing how
delay-robustness may depend on the subset of events subject to
delay, and that a given event may be delay-critical for some choices
of the delayed event subset but not for others.

With the definitions and tests reported here as basic tools, future
work should include the investigation of alternative channel models
and, of especial interest, global interconnection properties of a
distributed system of DES which render delay-robustness more or less
likely to be achieved. A quantitative approach involving timed
discrete-event systems could also be an attractive extension.

\appendices

\section{Proof of Proposition~\ref{pro1}}
\label{appA}


Recall that ${\bf SUP}' = (Y, \Sigma', \eta, y_0, Y_m)$. According to natural projection $P: \Sigma'^* \rightarrow \Sigma^*$ which maps $(\Sigma' - \Sigma)$ to $\epsilon$, define $\eta':Y \times \Sigma^* \rightarrow Pwr(Y)$ given by
\begin{align}\label{eqa1}
\eta'(y, t) = \{\eta(y,s)|s\in \Sigma'^*, \eta(y,s)! ~\& Ps = t\}.
\end{align}
Let $\rho$ be the supremal quasi-congruence on $Y$ with respect to ${\bf SUP}'$, and define $P_{\rho}: Y \rightarrow Y/\rho = \overline{Y}$. As in (\cite{Wonham:2011a}, Chapt. 6), ${\bf QCSUP}' = (\overline{Y}, \Sigma, \overline{\eta}, \overline{y}_0, \overline{Y}_m)$ is defined with $\overline{\eta}: \overline{Y} \times \Sigma^* \rightarrow Pwr(\overline{Y})$ given by
\begin{align} \label{eqa2}
\overline{\eta}(\overline{y}, t):= \bigcup \{P_\rho(\eta'(y,t))| P_\rho(y) = \overline{y}\},
\end{align}
$\overline{y}_0 = P_\rho(y_0)$ and $\overline{Y}_m = P_\rho(Y_m)$.

\begin{proof}
We must prove that ${\bf QCSUP}'$ represents $PL_m({\bf SUP}')$ and is a canonical generator.

(1) We show that ${\bf QCSUP}'$ represents $PL_m({\bf SUP}')$, i.e, \[L_m({\bf QCSUP}') = PL_m({\bf SUP}')\] and \[L({\bf QCSUP}') = PL({\bf SUP}').\]

(\rmnum{1}) $L({\bf QCSUP}') \subseteq PL({\bf SUP}')$

Let $t\in L({\bf QCSUP}')$. We prove by induction that $t \in PL({\bf SUP}')$.

{\bf Base step}: $t = \epsilon \in PL({\bf SUP}')$ trivially.

{\bf Inductive step}: Suppose $t \in L({\bf QCSUP}')$, $t\in PL({\bf SUP}')$, and $t\alpha \in L({\bf QCSUP}')$; we must prove  $t\alpha \in PL({\bf SUP}')$.

Since $t\alpha \in L({\bf QCSUP}')$, we have $\overline{\eta}(\overline{y_0}, t) !$ and $\overline{\eta}(\overline{y_0}, t\alpha) !$. So, $(\exists \overline{y} \in \overline{Y})\ \overline{y} = \overline{\eta}(\overline{y_0},t) \ \&\ \overline{\eta}(\overline{y},\alpha)!$. We have $\overline{y_0} = P_{\rho}y_0$. Since $t\in PL({\bf SUP}')$, $(\exists s\in L({\bf SUP'}))\ Ps = t$, i.e. $\eta(y_0, s)!$. So, $\eta(y_0, s) \in \eta'(y_0, t)$, i.e., $\eta'(y_0, t) \neq \emptyset$. Thus, $\overline{y} = P_{\rho}\eta'(y_0, t)$ because ${\bf QCSUP}'$ is deterministic. Since $\overline{\eta}(\overline{y},\alpha)!$ and $\eta'(y_0, t) \neq \emptyset$, there exists $y \in \eta'(y_0, t)$ such that $\overline{\eta}(\overline{y},\alpha) = P_\rho\eta'(y, \alpha)$. Hence, $\eta'(y_0, t\alpha) !$. However, according to (\ref{eqa1})
\[\eta'(y_0, t\alpha) = \{\eta(y_0, s)| s \in \Sigma^*, \eta(y_0, s)!, Ps = t\alpha\}.\]
Thus, $(\exists s \in L({\bf SUP}'))\ Ps = t\alpha$, so $t\alpha \in PL({\bf SUP}')$.

(\rmnum{2}) $PL({\bf SUP}') \subseteq L({\bf QCSUP}')$

Let $t\in PL({\bf SUP}')$; we show that $t \in L({\bf QCSUP}')$.

{\bf Base step}: $t = \epsilon \in L({\bf QCSUP}')$ trivially.

{\bf Inductive step}: Supposing $t \in PL({\bf SUP}')$, $t\in L({\bf QCSUP}')$, and $t\alpha \in PL({\bf SUP}')$, we show $t\alpha \in L({\bf QCSUP}'))$.

Since $t \in PL({\bf SUP}')$ and $t \in L({\bf QCSUP}')$, $\eta'(y_0, t) \neq \emptyset$, $\overline{\eta}(\overline{y_0},t) !$; letting $\overline{y} = \overline{\eta}(\overline{y_0},t)$, then $\overline{y} = P_{\rho}\eta'(y_0, t)$ because ${\bf QCSUP}'$ is deterministic. Since $t\alpha \in PL({\bf SUP}')$, there exists $s'\in L({\bf SUP}')$, i.e. $\eta(y_0, s')!$ such that $Ps' = t\alpha$; thus
\begin{align*}
&\bigcup \{\eta'(y',\alpha)|y' \in \eta'(y_0, t)\}\\
=~&\bigcup \{\eta'(y',\alpha)| s\in \Sigma'^*, y' = \eta (y_0, s), Ps = t\}~~ (\text{according to (\ref{eqa1})})\\
=~& \{\eta((\eta(y_0, s), v))| v\in \Sigma'^*, \eta(\eta(y_0, s), v)!, Ps = t, Pv = \alpha\} \\
=~& \{\eta(y_0, sv)| sv \in \Sigma'^*, \eta(y_0, sv)!, P(sv) = t\alpha\}\\
\neq~& \emptyset ~~(\text{since $\eta(y_0, s')!$ and $Ps' = t\alpha$}),
\end{align*}
i.e. there exists $y \in \eta'(y_0, t)$ such that $\eta'(y, \alpha)!$. Then, $P_{\rho}y = \overline{y}$ due to $\overline{y} = P_{\rho}\eta'(y_0, t)$. Hence, $\overline{\eta}(\overline{y}, \alpha) = P_{\rho}\eta'(y, \alpha) \neq \emptyset$, i.e., $\overline{\eta}(\overline{y}, \alpha)!$. So, $t\alpha \in L({\bf QCSUP}')$.

(\rmnum{3}) $L_m({\bf QCSUP}') \subseteq PL_m({\bf SUP}')$

For any $t\in \Sigma^*$, if $t\in L_m({\bf QCSUP}')$, then $(\exists \overline{y} \in \overline{Y})\ \overline{y} = \overline{\eta}(\overline{y}_0, t) \ \&\ \overline{y} \in \overline{Y}_m$. By (i), we conclude that $t \in PL({\bf SUP}')$. Thus, $\eta'(y_0, t) \neq \emptyset$. Because ${\bf QCSUP}'$ is deterministic, we know that $\overline{y} = P_{\rho}\eta'(y_0, t)$. So, $P_{\rho}\eta'(y_0, t) \in \overline{Y}_m$. Further, $\eta'(y_0, t) \cap Y_m \neq \emptyset$, i.e., there exists $s \in \Sigma'^*$ such that $\eta(y_0, s)!\ \&\ \eta(y_0, s) \in Y_m \ \&\ Ps = t$. Hence, $s \in L_m({\bf SUP}')$, thus $t = Ps \in PL_m({\bf SUP}')$.

(\rmnum{4}) $PL_m({\bf SUP}') \subseteq L_m({\bf QCSUP}')$

For any $t\in \Sigma^*$, if $t\in PL_m({\bf SUP}')$, then $\eta'(y_0, t)! \ \&\ \eta'(y_0, t) \cap Y_m \neq \emptyset$. By (ii), $t \in L({\bf QCSUP}')$, i.e., $(\exists \overline{y} \in \overline{Y})\ \overline{\eta}(\overline{y}_0, t)! \ \&\ \overline{y} = \overline{\eta}(\overline{y}_0, t)$. Since ${\bf QCSUP}'$ is deterministic, $\overline{y} = P_{\rho}\eta'(y_0,t)$. We conclude that $P_{\rho}\eta'(y_0,t) \in \overline{Y}_m$ from $\eta'(y_0, t) \cap Y_m \neq \emptyset$. Hence, $\overline{y} \in \overline{Y}_m$, i.e., $t\in L_m({\bf QCSUP}')$.

2. We prove that ${\bf QCSUP}'$ is a canonical(minimal-state) generator.

Let $\nu$ be a congruence on $\overline{Y}$ defined according to: $\overline{y} \equiv \overline{y'}$ (mod $\nu$) provided

(\rmnum{1}) ($\forall t \in \Sigma^*$) $\overline{\eta}(\overline{y}, t)! \Leftrightarrow \overline{\eta}(\overline{y'},t)!$

(\rmnum{2})($\forall t \in \Sigma^*$) $\overline{\eta}(\overline{y}, t)\in \overline{Y}_m \Leftrightarrow \overline{\eta}(\overline{y'},t) \in \overline{Y}_m$.

With reference to (\cite{Wonham:2011a}, Proposition 2.5.1), projection (mod $\nu$) reduces ${\bf QCSUP}'$ to a state-minimal generator.

Define $P_{\nu}:\overline{Y} \rightarrow \overline{Y}/\nu$ and write $\nu\circ\rho = ker(P_{\nu}\circ P_{\rho})$. Next we will prove that $\nu\circ\rho$ is a quasi-congruence on $Y$,i.e., for all $y,y' \in Y$,
\begin{align*}
P_{\nu}\circ P_{\rho}(y) &= P_{\nu}\circ P_{\rho}(y') \Rightarrow (\forall \alpha \in \Sigma)P_{\nu}\circ P_{\rho}\eta(y,\alpha) = P_{\nu}\circ P_{\rho}\eta(y',\alpha).
\end{align*}
Now
\begin{align*}
&P_{\nu}\circ P_{\rho}(y) = P_{\nu}\circ P_{\rho}(y') \\
\Rightarrow  ~&P_{\nu}(P_{\rho}(y)) = P_\nu(P_{\rho}(y')) \\
\Rightarrow ~&P_{\nu}(\overline{\eta}(P_{\rho}(y)), \alpha) = P_{\nu}(\overline{\eta}(P_{\rho}(y')), \alpha) \\
~&\text{  (cf. (ii) of Proposition 2.5.1 in \cite{Wonham:2011a})} \\
\Rightarrow ~&P_{\nu}(\overline{\eta}(\overline{y}, \alpha)) = P_\nu(\overline{\eta}(\overline{y'}, \alpha)) \\
\Rightarrow ~&P_\nu(P_{\rho}(\eta'(y,\alpha))) = P_\nu(P_{\rho}(\eta'(y',\alpha))) \\
\Rightarrow ~&P_\nu\circ P_{\rho}\eta'(y, \alpha) = P_\nu\circ P_{\rho}\eta'(y', \alpha)
\end{align*}

Hence, $\nu\circ\rho$ is a quasi-congruence on $Y$. Obviously, $\nu\circ\rho$ is coarser than $\rho$. However, $\rho$ is the supremal quasi-congruence on $Y$, so for any $y,y' \in Y$, if $P_\nu(P_{\rho}(y)) = P_\nu(P_{\rho}(y'))$, i.e., $(y,y')\in \nu\circ\rho$, then $(y, y') \in \rho$, which means that $P_{\rho}(y) = P_{\rho}(y')$. Hence, $\nu = \bot$ (namely all its cells are singletons).

We have shown that ${\bf QCSUP}'$ is a canonical generator.
\end{proof}

\section{Proof of Proposition~\ref{pro:relchn}}\label{app0}


For the proof, we need the natural projections:
\begin{align*}
Q':&\Sigma'^* \rightarrow \Sigma^* \\
Q_{T}':&\Sigma_T'^* \rightarrow \Sigma^*\\
Q_{r_{12}'}:&\Sigma_T'^* \rightarrow (\Sigma\cup\{r_{21}'\})^*\\
Q_{r_{21}'}:&(\Sigma\cup\{r_{21}'\})^* \rightarrow \Sigma^*\\
Q_{ch}:&\Sigma'^*\rightarrow \{r,r'\}^*\\
Q_{Tch}:&\Sigma_T'^*\rightarrow \{r,r_{21}',r_{12}'\}^*.
\end{align*}
Thus $Q_{T}' = Q_{r_{21}'}Q_{r_{12}'}$. According to the definition
of ${\bf CH}(2,r,1)$ and ${\bf TCH}(2,r,1)$, $L({\bf CH}(2,r,1)) = \overline{(r.r')^*}$
and  $L({\bf TCH}(2,r,1)) = \overline{(r.r_{21}'.r_{12}')^*}$.

Let ${\bf NSUP} = Sync({\bf SUP}_1', {\bf SUP}_2)$; then
\begin{subequations}
\begin{align}
L({\bf SUP}') &= L({\bf NSUP})\cap Q_{ch}^{-1}L({\bf CH}(2,r,1)), \label{eq:NSUP1}\\
L_m({\bf SUP}') &= L_m({\bf NSUP})\cap Q_{ch}^{-1}L_m({\bf CH}(2,r,1)).\label{eq:NSUP2}
\end{align}
\end{subequations}
Let ${\bf TNSUP} = Sync({\bf TSUP}_1', {\bf SUP}_2)$; then
\begin{subequations}
\begin{align}
L({\bf TSUP}') &= Q_{r_{12}'}^{-1}L({\bf TNSUP})\cap Q_{Tch}^{-1}L({\bf TCH}(2,r,1)),\label{eq:TNSUP1}\\
L_m({\bf TSUP}') &= Q_{r_{12}'}^{-1}L_m({\bf TNSUP})\cap Q_{Tch}^{-1}L_m({\bf TCH}(2,r,1)). \label{eq:TNSUP2}
\end{align}
\end{subequations}

Since from $\bf NSUP$ (resp. $\bf TNSUP$) to $\bf TNSUP$ (resp. $\bf NSUP$), only $r'$
(resp. $r_{21}'$) is replaced by $r_{21}'$ (resp. $r'$), we still have the following results:
\begin{subequations}
\begin{align}
&s = x_1.r.x_2 \in L({\bf NSUP}) \Leftrightarrow t = x_1.r.x_2 \in L({\bf TNSUP}) \label{eq:relr1}\\
&s = x_1.r.x_2.r'.x_3 \in L({\bf NSUP}) \Leftrightarrow t = x_1.r.x_2.r_{21}'.x_3 \in L({\bf TNSUP}) \label{eq:relr2}
\end{align}
\end{subequations}
where the strings $x_1, x_2$, and $x_3$ are free of $r,r'$ and $r_{21}'$.
Furthermore,
\begin{align}\label{eq:ch}
Q'L({\bf SUP}') &= Q'\big(L({\bf NSUP})\cap Q_{ch}^{-1}L({\bf CH}(2,r,1))\big) \notag \\
                    &= Q'\big(L({\bf NSUP})\cap \overline{((\Sigma-\{r\})^*.r(\Sigma-\{r\})^*.r')^*}\notag\\
                    &= Q_{r_{21}'}\big(L({\bf TNSUP})\cap \overline{((\Sigma-\{r\})^*.r(\Sigma-\{r\})^*.r_{21}')^*}\big)\\
                    &~~~~~~~~~~~\mbox{~~(From $\bf NSUP$ to $\bf TNSUP$, $r'$ is replaced by $r_{21}'$)}\notag\\
                    &= Q_{r_{21}'}\big(L({\bf TNSUP}) \cap Q_{r_{12}'}(Q_{Tch}^{-1}L({\bf TCH}(2,r,1)))\big) \notag
\end{align}

Also, we need the following lemmas.
\begin{lemma} ($r'$, $r_{21}'$ and $r_{12}'$ \emph{insertion}) \label{lem:rinsert}
Let $s=x_1.r.x_2 \in L({\bf SUP})$ where the strings $x_1, x_2$ are free of $r$;
then $s' = x_1.r.r'.x_2 \in L({\bf SUP}')$, and $t' = x_1.r.r_{21}'.r_{12}'.x_2 \in L({\bf TSUP}')$.
\end{lemma}

\emph{Proof.} Immediate from the definition of relevant synchronous product.

\begin{lemma} \label{lem:ack}
Let $s' = x_1.r.x_2.r'.x_3 \in L_m({\bf SUP}')$, where the strings $x_i(i = 1,2,3)$ are
free of $r,r'$. For any $x_{31}, x_{32} \in (\Sigma-\{r\})^*$ that satisfy $x_3 = x_{31}.x_{32}$,
$t' := x_1.r.x_2.r_{21}'.x_{31}.r_{12}'.x_{32} \in L_m({\bf SUP}'')$. On the other side,
if $t' = x_1.r.x_2.r_{21}'.x_{31}.r_{12}'.x_{32} \in L_m({\bf SUP}'')$, then $s' = x_1.r.x_2.r'.x_{31}.x_{32}
\in L_m({\bf SUP}')$.
\end{lemma}

\emph{Proof.} For the first part, it follows from $s' \in L_m({\bf SUP}') = L_m({\bf NSUP})\cap Q_{ch}^{-1}L_m({\bf CH}(2,r,1))$
that $x_1.r.x_2.r'.x_3 \in L_m({\bf NSUP})$. By (\ref{eq:relr2}), $x_1.r.x_2.r_{21}'.x_3 \in L_m({\bf TNSUP})$.
So $Q_{r_{12}'}t' = x_1.r.x_2.r_{21}'.x_{31}.\\x_{32} \in L_m({\bf TNSUP})$, and thus $t' \in Q_{r_{12}'}^{-1}L_m({\bf TNSUP})$.
Furthermore, $Q_{Tch}t' = r.r_{21}'.r_{12}' \in L_m({\bf TCH}(2,\\r,1))$. Hence,
$t' \in Q_{r_{12}'}^{-1}L_m({\bf TNSUP}) \cap Q_{Tch}^{-1}L_m({\bf TCH}(2,r,1) = L_m({\bf TSUP}')$.
The argument for the second part is similar.

\emph{Proof of Proposition~\ref{pro:relchn}.} (If) We assume that
\begin{subequations} \label{TSUP'}
\begin{align}
& Q'L({\bf SUP}')=L({\bf SUP}) \label{SUPr'a} \\
& Q'L_m({\bf SUP}')=L_m({\bf SUP}) \label{SUPr'b}\\
& Q' \mbox{ has the observer property with respect to ${\bf SUP}'$
and ${\bf SUP}$}  \label{SUPr'c}.
\end{align}
\end{subequations}
It must be shown that the counterpart properties hold for $Q_T'$ and
${\bf TSUP}'$, namely
\begin{subequations} \label{TSUP'}
\begin{align}
& Q_{T}'L({\bf TSUP}')=L({\bf SUP}) \label{TSUP'a} \\
& Q_{T}'L_m({\bf TSUP}')=L_m({\bf SUP}) \label{TSUP'b} \\
& Q_{T}' \mbox{ has the observer property with respect to ${\bf TSUP}'$
and ${\bf SUP}$}  \label{TSUP'c} .
\end{align}
\end{subequations}

For $(\subseteq)$ of (\ref{TSUP'a}),
\begin{align*}
Q_{T}'L({\bf TSUP}') &= Q_{T}' \big(Q_{r_{12}'}^{-1}L({\bf TNSUP})\cap Q_{Tch}^{-1}L({\bf TCH}(2,r,1))\big)\\
                     &= (Q_{r_{21}'} Q_{r_{12}'}) \big(Q_{r_{12}'}^{-1}L({\bf TNSUP})\cap Q_{Tch}^{-1}L({\bf TCH}(2,r,1))\big)\\
                  &\subseteq Q_{r_{21}'}\big(L({\bf TNSUP}) \cap Q_{r_{12}'}(Q_{Tch}''^{-1}L({\bf TCH}(2,r,1)))\big)\\
                  & = Q'L({\bf SUP}')~~~\mbox{(By (\ref{eq:ch}))}\\
                  & \subseteq L({\bf SUP}). ~~~\mbox{(By (\ref{SUPr'a}))}
\end{align*}
For $(\supseteq)$ of (\ref{TSUP'a}), if $s=x_1.r.x_2 \in L({\bf SUP})$, then
applying Lemma~\ref{lem:rinsert} to $s$ with $r_{21}'$ and $r_{12}'$ we get
that $t'=x_1.r.r_{21}'.r_{12}'.x_2 \in L({\bf TSUP}')$ and then
$s=Q_T'(t')$ , as claimed. The argument for (\ref{TSUP'b}) is similar.

For the observer property we have by (\ref{SUPr'c}) that
\begin{align*}
(\forall s' \in L({\bf SUP}'))(\forall v \in \Sigma^*)& Q'(s').v \in
L_m({\bf SUP})  \Rightarrow\\
&(\exists v' \in (\Sigma')^*) s'.v' \in L_m({\bf SUP}') \ \&\
Q'(v')=v
\end{align*}
and must verify the counterpart (\ref{TSUP'c}), namely
\begin{align*}
(\forall t' \in L({\bf TSUP}'))(\forall u \in \Sigma^*)& Q_T'(t').u
\in L_m({\bf SUP})  \Rightarrow\\
&(\exists u' \in (\Sigma_T')^*) t'.u' \in L_m({\bf TSUP}') \ \&\
Q_T'(u')=u.
\end{align*}
For the proof let $t' \in L({\bf TSUP}')$, $u \in \Sigma^*$, $Q_T'(t').u
\in L_m({\bf SUP})$. Next we prove (\ref{TSUP'c}) from the following three
cases: (1) $t' = x_1.r.x_2$, (2)$t' = x_1.r.x_2.r_{21}'.x_3$ and (3)$t' =
x_1.r.x_2.r_{21}'.x_3.r_{12}'.x_4$, where $x_i(i = 1,2,3,4)$ are free of
$r$, $r_{21}'$, and $r_{21}'$. Note that since the re-transmission of $r$
will not start until the last transmission is completed, in this proof we only
consider the transmission of one instance of $r$.

(1) By $t' \in L({\bf TSUP}')$, we have $t' \in Q_{r_{12}'}^{-1}L({\bf TNSUP})$.
Since $t'$ is free of $r_{12}'$, $x_1.r.x_2 = Q_{r_{12}}t' \in L({\bf TNSUP})$.
By (\ref{eq:relr1}), $x_1.r.x_2 \in L({\bf NSUP})$. Also, $Q_{ch}(x_1.r.x_2) = r
\in L({\bf CH}(2,r,1))$. So, $s':= x_1.r.x_2 \in L({\bf NSUP})\cap Q_{ch}^{-1}
L({\bf CH}(2,r,1)) = L({\bf SUP}')$. Define $v = u$; then $Q'(s').v = Q_T'(t').u \in
L_m({\bf SUP})$. By (\ref{SUPr'c}), there exists $v' \in \Sigma'^*$ such that
$Q'v' = v$ and $s'.v' \in L_m({\bf SUP}')$, i.e. $x_1.r.x_2.v \in L_m({\bf SUP}')$.
By (\ref{eq:NSUP2}), $s'.v' \in Q_{ch}^{-1}L_m({\bf CH}(2,r,1))$; thus $v'$ can
be written as $v_1'.r'.v_2'$ where $v_1'$ and $v_2'$ are free of $r'$.
Namely, $x_1.r.x_2.v_1'.r'v_2' \in L_m({\bf SUP}')$. By Lemma~\ref{lem:ack},
$x_1.r.x_2.v_1'.r_{21}'.r_{12}'.v_2'\in L_m({\bf TSUP}')$.
Define $u' = v_1'.r_{21}'.r_{12}'.v_2'$; then $Q_T'u' = v_1'v_2'
= Q'v' = v = u$, and $t'.u' \in L_m({\bf TSUP}')$, as required by (\ref{TSUP'c}).

(2) Similar to case (1), we have $t' \in L({\bf TNSUP})$. By (\ref{eq:relr2}),
$s' := x_1.r.x_2.r'.x_3 \in L({\bf NSUP})$. Furthermore, since $Q's' = r.r' \in
L({\bf CH}(2,r,1))$, $s' \in L({\bf SUP}')$. Define $v = u$; then $Q'(s').v = Q_T(t').u
\in L_m({\bf SUP})$. By (\ref{SUPr'c}), there exists $v' \in \Sigma'^*$ such that
$Q'v' = v$ and $s'.v' \in L_m({\bf SUP}')$.
By (\ref{eq:NSUP2}), $s'.v' \in Q_{ch}^{-1}L_m({\bf CH}(2,r,1))$; thus $v'$ is free of
$r'$, i.e. $v' = v$ (In this proof only one instance of $r$ is taken into consideration).
So, $x_1.r.x_2.r'.x_3.v' \in L_m({\bf SUP}')$. By Lemma~\ref{lem:ack},
$x_1.r.x_2.r_{21}'.x_3.r_{12}'.v' \in L_m({\bf TSUP}')$. Define $u' = r_{12}'.v'$; then
$Q_{T}'u' = v' = v = u$ and $t'.u' \in L_m({\bf TSUP}')$, as required by (\ref{TSUP'c}).

(3) Let $s' := x_1.r.x_2.r'.x_3.x_4$. By (\ref{eq:TNSUP1}), we have $s' = Q_{r_{12}'}t'
\in L({\bf TNSUP})$. Similar to case (2), if defining $v' = u$, then we can verify
that $x_1.r.x_2.r_{21}'.x_3.r_{12}'x_4.v' \in L_m({\bf TSUP}')$. Define $u' = v'$; then
$Q_{T}'u' = v' = u$ and $t'u' \in L_m({\bf TSUP}')$, as required by (\ref{TSUP'c}).

\vspace{1em}

(Only if) We assume that conditions (\ref{TSUP'a})-(\ref{TSUP'c}) hold; it must be
shown that conditions (\ref{SUPr'a})-(\ref{SUPr'c}) hold.

For $(\subseteq)$ of (\ref{SUPr'a}), let $s' \in L({\bf SUP}')$; we prove that
$Q's \in L({\bf SUP})$ from the following two cases: (1) $s' = x_1.r.x_2$, and
(2) $s' = x_1.r.x_2.r'.x_3$, where $x, x_1, x_2, x_3$ are free of $r$ and $r'$.


(1) It follows from $s' \in L({\bf SUP}')$ that $x_1.r.x_2 \in L({\bf NSUP})$.
By (\ref{eq:relr1}), we have $t:= x_1.r.x_2 \in L({\bf TNSUP})$, and thus
$t \in Q_{r_{12}'}^{-1}L({\bf TNSUP})$. Also, $Q_{Tch}t = r \in L({\bf TCH}(2,r,1))$.
So, $t \in L({\bf TSUP}')$, and thus $Q_T't \in Q_T'L({\bf SUP}') \subseteq L({\bf SUP})$.
Hence, we also have $Q's' = t = Q_T't \in L({\bf SUP})$.

(2) Similar to case (1), we have $x_1.r.x_2.r'.x_3 \in L({\bf NSUP})$. By (\ref{eq:relr2}),
$t:= x_1.r.x_2.r_{21}'.x_3 \in L({\bf TNSUP})$. Let $t':= x_1.r.x_2.r_{21'}.x_3.r_{12}'$;
then $t' \in Q_{r_{12}'}^{-1}L({\bf TNSUP})$. Also, $Q_{Tch}t' = r.r._{21}'.r_{12}' \in
L({\bf TCH}(2,r,1))$. So, $t' \in L({\bf TSUP}')$, and thus $Q_T't' \in Q_T'L({\bf TSUP}')
\subseteq L({\bf SUP})$. Hence, $Q's' = x_1.r.x_2.x_3 = Q_T't' \in L({\bf SUP})$.

$(\supseteq)$ of (\ref{SUPr'a}) can be verified similar to the proof of ($\supseteq$)
of (\ref{TSUP'a}). The argument for (\ref{SUPr'b}) is similar.

For the observer property we have by (\ref{TSUP'c}) that
\begin{align*}
(\forall t' \in L({\bf TSUP}'))(\forall u \in \Sigma^*)& Q_T'(t').u
\in L_m({\bf SUP})  \Rightarrow\\
&(\exists u' \in (\Sigma_T')^*) t'.u' \in L_m({\bf TSUP}') \ \&\
Q_T'(u')=u
\end{align*}
and must verify the counterpart (\ref{SUPr'c}), namely
\begin{align*}
(\forall s' \in L({\bf SUP}'))(\forall v \in \Sigma^*)& Q'(s').v \in
L_m({\bf SUP})  \Rightarrow\\
&(\exists v' \in (\Sigma')^*) s'.v' \in L_m({\bf SUP}') \ \&\
Q'(v')=v.
\end{align*}
For the proof let $s' \in L({\bf SUP}')$, $v \in \Sigma^*$, $Q'(s').v
\in L_m({\bf SUP})$. Next we prove (\ref{SUPr'c}) from the following two
cases: (1) $s' = x_1.r.x_2$, (2)$s' = x_1.r.x_2.r'.x_3$, where $x_i(i = 1,2,3)$
are free of $r$ and $r'$.

(1) Similar to case (1) in proving ($\subseteq$) of (\ref{SUPr'a}), by $s' \in L({\bf SUP}')$,
we have $t':= x_1.r.x_2 \in L({\bf TSUP}')$. Define $u = v$; then $Q_T'(t').u = Q'(s').v
\in L_m({\bf SUP})$. By ({\ref{TSUP'c}}), there exists $u' \in \Sigma_T'$ such that
$Q_T'u' = u$ and $t'.u' \in L_m({\bf TSUP}')$. Namely, $x_1.r.x_2.u' \in L_m({\bf TSUP}')$.
So by $Q_{Tch}(x_1.r.x_2.u') = r.Q_{Tch}(u')$ there must exist $u_1',u_2',u_3' \in
\Sigma^*$ such that $u' = u_1'.r_{21}'.u_2.r_{12}'.u_3$. Applying Lemma~\ref{lem:ack},
$x_1.r.x_2.u_1'.r'.u_2'.u_3' \in L_m({\bf SUP}')$. Define $v' = u_1'.r'.u_2'.u_3'$;
then $Q'v' = u_1'.u_2'.u_3' = Q_T'u' = u = v$, and $s'.v' \in L_m({\bf SUP}')$, as required
by (\ref{SUPr'c}).

(2) Similar to case (2) in proving ($\subseteq$) of (\ref{SUPr'a}), by
$s' \in L({\bf SUP}')$, we have $t': = x_1.r.x_2.r_{21}'.x_3.r_{12}' \in L({\bf TSUP}')$.
Define $u = v$; then $Q_T'(t').u = x_1.r.x_2.x_3.v = Q'(s')v \in L_m({\bf SUP})$.
By ({\ref{TSUP'c}}), there exists $u' \in \Sigma_T'$ such that $Q_T'u' = u$ and $t'u'
\in L_m({\bf TSUP}')$. Namely, $x_1.r.x_2.r_{21'}.x_3.r_{12}'.u' \in L_m({\bf TSUP}')$.
Since $Q_{Tch}(x_1.r.x_2.r_{21'}.x_3.r_{12}'.u') = (r.r_{21}'.r_{12}').Q_{Tch}(u')$,
and only one instance of $r$ is taken into consideration, $u'$ is free of $r_{21}'$,
and $r_{12}'$ (also $u'$ is free of $r'$); thus $Q_T'u' = u' = Q'u'$. Applying Lemma~\ref{lem:ack},
we obtain that $x_1.r.x_2.r'.x_3.u' \in L_m({\bf SUP}')$. Define $v' = u'$;
then $Q'v' = Q'u' = u' = Q_T'u' = u = v$, and $s'.v' \in L_m({\bf SUP}')$, as required
by (\ref{SUPr'c}).


\section{Delay-Robustness of Decentralized Controllers}
\label{appB}


Here we show that the verification tool for delay-robustness of distributed controllers can be used without change to verify the delay-robustness of decentralized supervisors.

Let $\bf G$ be the DES to be controlled, and ${\bf LOC}_1$ and ${\bf LOC}_2$ be two decentralized controllers, which achieve global supervision with zero-delay communication. Let $\Sigma_i$, $\Sigma_{io}$ be the event set and observable event set of ${\bf LOC}_i$, respectively $(i = 1, 2)$. Assume event $r \in \Sigma_1 \cap (\Sigma_{2o} - \Sigma_{1o})$, which is not observed by ${\bf LOC}_1$, but is observed by ${\bf LOC}_2$. Hence, $r$ should be transmitted to ${\bf LOC}_1$. We use the channel ${\bf CH}(2,r,1)$, as shown in Fig.~\ref{fig1}, to transmit $r$ and use $r'$ to represent that ${\bf LOC}_1$ receives the occurrence of $r$. Then, replacing $r$ by $r'$, we obtain ${\bf LOC}_1'$. Let ${\bf SUP} = Sync({\bf G}, {\bf LOC}_1, {\bf LOC}_2)$, ${\bf SUP}' = Sync({\bf G}, {\bf LOC}_1', {\bf CH}(2,r,1), {\bf LOC}_2)$, and ${\bf QCSUP}' = Supqc({\bf SUP}', Null[r'])$. Finally, by  Theorem~\ref{thm1}, if ${\bf SUP} \approx {\bf QCSUP}'$, $\bf SUP$ is delay-robust with respect to $r$, or ${\bf LOC}_1$ and ${\bf LOC}_2$ achieve global supervision with unbounded delay communication.

\section{Proof of Theorem~\ref{thm2}}
\label{appC}


The relevant natural projections are
\begin{align*}
P':&(\Sigma_1 \cup \{\alpha',\beta'\} \cup \Sigma_{E'})^* \rightarrow \Sigma^*\\
P'':&(\Sigma_1 \cup \{\beta'\} \cup \Sigma_{E'})^* \rightarrow
\Sigma^*.
\end{align*}
Thus $P'$ (resp. $P''$) nulls $\{\alpha',\beta'\}$ (resp.
$\{\beta'\}$) $\cup \{r' | r' \in \Sigma_{E'}\}$.

For the proof we assume that
\begin{subequations} \label{SUP'}
\begin{align}
& P'L({\bf SUP}')=L({\bf SUP}) \label{SUP'a} \tag{\theequation a}\\
& P'L_m({\bf SUP}')=L_m({\bf SUP}) \label{SUP'b} \tag{\theequation b}\\
& P' \mbox{ has the observer property with respect to ${\bf SUP}'$
and ${\bf SUP}$}  \label{SUP'c} \tag{\theequation c}.
\end{align}
\end{subequations}
It must be shown that the counterpart properties hold for $P''$ and
${\bf SUP}''$, namely
\begin{subequations} \label{SUP''}
\begin{align}
& P''L({\bf SUP}'')=L({\bf SUP}) \label{SUP''a} \tag{\theequation a}\\
& P''L_m({\bf SUP}'')=L_m({\bf SUP}) \label{SUP''b} \tag{\theequation b}\\
& P'' \mbox{ has the observer property with respect to ${\bf SUP}''$
and ${\bf SUP}$}  \label{SUP''c} \tag{\theequation c}.
\end{align}
\end{subequations}

We need the following lemmas.
\begin{lemma} ($\alpha'$ \emph{insertion}) \label{lem:1}
Let $s''=x.\alpha.x.\beta.x.\beta'.x \in L({\bf SUP}'')$ where
the (generally distinct) strings written $x$ are free
of $\alpha$, $\beta$, $\beta'$. Then
$s' := x.\alpha.\alpha'.x.\beta.x.\beta'.x \in L({\bf SUP}')$.
\end{lemma}

\emph{Proof.} Immediate from the definition of the relevant
synchronous products. \hfill $\square$

Evidently Lemma~\ref{lem:1} extends to multiple appearances of
$\alpha$, $\beta$, $\beta'$ and arbitrary possible orderings of the
$\alpha$ with respect to the $\beta$, $\beta'$; and holds with $L$
replaced by $L_m$ throughout.

\begin{lemma} ($\alpha'$ \emph{deletion}) \label{lem:2}
Let $t' = x.\alpha.y.\alpha'.z.\beta.z.\beta'.z \in L_m({\bf
SUP}')$, where the strings $x,y,z$ are free of $\alpha$, $\alpha'$,
$\beta$, $\beta'$. Then $t'' := x.\alpha.y.z.\beta.z.\beta'.z \in
L_m({\bf SUP}'')$.
\end{lemma}

\emph{Proof.} Recall that the synchronous products defining
$L_m({\bf SUP}')$ and $L_m({\bf SUP}'')$ differ only in that the
latter omits the factor ${\bf CH}(E,\alpha,1)$, and in ${\bf
SUP}''_1$ $\alpha$ appears as in ${\bf SUP}_1$ (and not as
$\alpha'$).  The string $y$ is of form, say $a_1.b_1.a_2.b_2$, where
$a_1,a_2 \in (\Sigma'_1)^*$ and $b_1,b_2 \in \Sigma_{E'}^*$, hence
by definition of synchronous product can be re-ordered as
$a_1.a_2.b_1.b_2$ without affecting membership of $t'$ in $L_m({\bf
SUP}')$; next $\alpha.y$ can be re-ordered in $t'$ as
$a_1.a_2.\alpha.b_1.b_2$, and then $\alpha.y.\alpha'$ can be
re-ordered as $a_1.a_2.\alpha.\alpha'.b_1.b_2$, again preserving
membership of $t'$ in $L_m({\bf SUP}')$. In this new ordering it is
clear that deletion of $\alpha'$ converts $t'$ to a string $t''$ in
$L_m({\bf SUP}'')$. Reversing the ordering restores our original
$t''$, proving the claim. \hfill $\square$

\emph{Proof of Theorem~\ref{thm2}.} For (\ref{SUP''a}) suppose
$s''=x.\alpha.x.\beta.x.\beta'.x \in L({\bf SUP}'')$. By
Lemma~\ref{lem:1}, $s':=x.\alpha.\alpha'.x.\beta.x.\beta'.x \in
L({\bf SUP}')$, so by (\ref{SUP'a}) $P'(s') \in L({\bf SUP})$.
Evidently $P''(s'')=P'(s')$ as required. For the reverse inclusion,
if $s=x.\alpha.x.\beta.x \in L({\bf SUP})$ then applying
Lemma~\ref{lem:1} to $s$ with $\beta$ we get that
$s''=x.\alpha.x.\beta.\beta'.x \in L({\bf SUP}'')$ and then
$s=P''(s'')$ , as claimed. The argument for (\ref{SUP''b}) is
similar. For the observer property we have by (\ref{SUP'c}) that
\begin{align*}
(\forall s' \in L({\bf SUP}'))(\forall v \in \Sigma^*)& P'(s').v \in
L_m({\bf SUP})  \Rightarrow\\
&(\exists v' \in (\Sigma')^*) s'.v' \in L_m({\bf SUP}') \ \&\
P'(v')=v
\end{align*}
and must verify the counterpart (\ref{SUP''c}), namely
\begin{align*}
(\forall s'' \in L({\bf SUP}''))(\forall v \in \Sigma^*)& P''(s'').v
\in L_m({\bf SUP})  \Rightarrow\\
&(\exists v'' \in (\Sigma'')^*) s''.v'' \in L_m({\bf SUP}'') \ \&\
P''(v'')=v.
\end{align*}
For the proof let $s'' \in L({\bf SUP}'')$, $v \in \Sigma^*$,
$P''(s'').v \in L_m({\bf SUP})$. By Lemma~\ref{lem:1} with
$\alpha'$-insertion we obtain $s' \in L({\bf SUP}')$ such that
$P'(s')=P''(s'')$, so $P'(s').v \in L_m({\bf SUP})$, and by
(\ref{SUP'c}) there is $v' \in (\Sigma')^*$ with $s'.v' \in L_m({\bf
SUP}')$ and $P'(v')=v$.  Thus $v'$ is of the form
$v'=y.\alpha.y.\alpha'.y.\beta.y.\beta'.y$ (possibly with multiple
$\alpha$'s and $\beta$'s in various interleavings). Define
$v''=Q(v')$ where $Q$ projects $\alpha'$ to the empty string
$\epsilon$. Then $P''(v'')=P''Q(v')=P'(v')=v$. Also, by
Lemma~\ref{lem:2}, $s''.v''=Q(s'.v') \in QL_m({\bf SUP}') \subseteq
L_m({\bf SUP}'')$. Thus $v''$ has the properties required in
(\ref{SUP''c}), which completes the proof. \hfill $\square$




\bibliographystyle{unsrt}        
\bibliography{zhang_references}

\end{document}